\documentclass[journal,10pt,twocolumn]{IEEEtran}
\usepackage{cite}
\usepackage{graphicx}  
\usepackage{tabulary}
\usepackage{amsfonts} 
\usepackage{amsmath}
\usepackage{chngpage}
\usepackage{makecell}
\usepackage{array}
\usepackage{amsthm}
\usepackage{booktabs}
\usepackage{amssymb}
\usepackage{multirow}
\usepackage{enumerate}
\usepackage{color}
\usepackage[dvipsnames]{xcolor}
\usepackage{mathrsfs}

\usepackage{float}
\usepackage{algorithm}  
\usepackage{algpseudocode}  
\usepackage{amsmath}  
\usepackage{url}

\usepackage{graphicx}
\usepackage{subfigure}
\usepackage[font=small]{caption}

\usepackage{mathtools}
\usepackage{enumitem}

\setlist[enumerate,1]{label=(\roman*)}

\newtheoremstyle{noparens}%
{}{}%
{\itshape}{}%
{\bfseries}{.}%
{ }%
{\thmname{#1}\thmnumber{ #2}\mdseries\thmnote{ #3}}

\newcommand{\bluemark}[1]{{\color{Blue}#1}}

\theoremstyle{noparens}
\newtheorem{theorem}{Theorem}

\newtheorem{symmetry assumption}{Symmetry Assumption}
\newtheorem{lemma}{Lemma}
\newtheorem{corollary}{Corollary}

\newtheorem{proposition}{Proposition}
\newtheorem{definition}{Definition}
\newtheorem{property}{Property}

\begin{document}

%\title{Semantic Communication in the Finite Blocklength Regime}
%\title{Integrated Data-Semantic Communication in the Finite Blocklength Regime}
%\title{Joint Shannon-Semantic Information Compression in the Finite Blocklength Regime}
\title{Joint Data and Semantics Lossy Compression: Nonasymptotic Converse Bounds and Second-Order Asymptotics}
%\title{Joint Data and Semantics Lossy Compression in the Finite Blocklength Regime}
%\title{Joint Data-Meaning Compression in the Finite Blocklength Regime}
%\author{Huiyuan Yang\\2023.07.25}

\author{
	\IEEEauthorblockN{Huiyuan Yang,~\IEEEmembership{Graduate Student Member,~IEEE}, Yuxuan Shi, Shuo Shao,~\IEEEmembership{Member,~IEEE}, Xiaojun Yuan,~\IEEEmembership{Senior Member,~IEEE}}
	\thanks{This work has been submitted to the IEEE for possible publication. Copyright may be transferred without notice, after which this version may no longer be accessible.}
	\thanks{Huiyuan Yang and Xiaojun Yuan \emph{(corresponding author)} are with the National Key Laboratory on Wireless Communications, University of Electronic Science and Technology of China, Chengdu 611731, China (e-mail: hyyang@std.uestc.edu.cn; xjyuan@uestc.edu.cn).}
	\thanks{Yuxuan Shi is with the School of Cyber and Engineering, Shanghai Jiao Tong University, Shanghai 200240, China, and also with ZGC Institute of Ubiquitous-X Innovation and Applications, Beijing 100876, China (e-mail: ge49fuy@sjtu.edu.cn).}
	\thanks{Shuo Shao \emph{(corresponding author)} is with the School of Cyber and Engineering, Shanghai Jiao Tong University, Shanghai 200240, China (e-mail: shuoshao@sjtu.edu.cn).}
}

\maketitle
\IEEEpeerreviewmaketitle

\begin{abstract}
%In this paper, we study the joint data and semantics lossy compression problem. Specifically, we consider data and its semantics are generated by two correlated sources and only the data source is observable.
%This paper studies the nonasymptotic and second-order properties of the joint data and semantics lossy compression problem.

This paper studies the joint data and semantics lossy compression problem, i.e., an extension of the hidden lossy source coding problem that entails recovering both the hidden and observable sources. We aim to study the nonasymptotic and second-order properties of this problem, especially the converse aspect. Specifically, we begin by deriving general nonasymptotic converse bounds valid for general sources and distortion measures, utilizing properties of distortion-tilted information. Subsequently, a second-order converse bound is derived under the standard block coding setting through asymptotic analysis of the nonasymptotic bounds. This bound is tight since it coincides with a known second-order achievability bound. We then examine the case of erased fair coin flips (EFCF), providing its specific nonasymptotic achievability and converse bounds. Numerical results under the EFCF case demonstrate that our second-order asymptotic approximation effectively approximates the optimum rate at given blocklengths.

\end{abstract}


\section{Introduction}
%\IEEEPARstart{I}{n} this paper, we consider a nonasymptotic joint data and semantics lossy compression problem proposed in our prior work \cite{Yang2024Joint}. More precisely, we consider that data and its semantics are generated by two correlated sources, $X$ and $S$, termed as the data and semantic sources, respectively. Note that $X$ and $S$ are general sources, not necessarily memoryless. As semantics are typically inferred from data and cannot be directly observed, we posit that the compressor or encoder can only access the data source $X$. Moreover, for practical considerations, we adopt a finite codebook. Therefore, ``nonasymptotic joint data and semantics lossy compression'' aims to encode $X$ into a codeword chosen from a finite codebook so that the decoder can recover both $S$ and $X$ based on it in a lossy sense. For a formal formulation, please refer to Section II.
\IEEEPARstart{D}{ata} and its semantics are two key objects of semantic communication;
%\IEEEPARstart{T}{wo} key objects of semantic communication are data and its semantics; 
the data is typically observable, while the underlying semantics can only be inferred from the data \cite{Qin2021Semantic, Shi2021From, Yang2023Semantic, Gunduz2023Beyond, Getu2023Making}. As the era unfolds and technology advances, there is a growing demand to generate data representations capable of recovering both the data and its semantics in a lossy sense, as outlined in \cite{Liu2021Rate, Liu2022Indirect, Huang2023Joint, Yang2024Joint}.
%As the era unfolds, an increasing demand emerges to simultaneously recover data and its semantics based on some lossy compressed representations, as outlined in \cite{Liu2021Rate, Liu2022Indirect, Huang2023Joint, Yang2024Joint}.
%As the era unfolds, there is a growing demand for simultaneous lossy compression of data and semantics, as outlined in \cite{Liu2021Rate, Liu2022Indirect, Huang2023Joint, Yang2024Joint}.
%As the era unfolds, there is an inevitable demand for simultaneous data and semantics lossy compression, as outlined in \cite{Liu2021Rate, Liu2022Indirect, Huang2023Joint, Yang2024Joint}. 
In order to address this challenge, a two-step approach might first come to mind, wherein data is first compressed/encoded without considering its semantics, followed by semantic inference based on the recovered data. However, this naive approach is generally suboptimal due to its semantic-agnostic encoding and distorted-data-based semantic inference. Actually, the optimal approach is to advance the semantic inference step to the compressor/encoder, i.e., to achieve joint compression of the data and its semantics by exploiting their known statistical relationship. In this paper, we investigate the problem of joint data and semantics lossy compression (JDSLC).
%Clearly, to achieve optimal performance, the compressor or encoder must comprehensively consider both data and semantics; in other words, joint compression is optimal. 

A block coding version of the JDSLC problem has been formulated in \cite{Liu2021Rate}, in which the asymptotic properties of this problem were investigated. Subsequently, in our prior work \cite{Yang2024Joint}, we defined a general version of the JDSLC problem, termed the general JDSLC problem in this paper, which allowed us to study the nonasymptotic behaviors of JDSLC. Note that in \cite{Yang2024Joint}, we studied the achievability aspect of the general JDSLC problem, presenting nonasymptotic and second-order achievability bounds. In this paper, instead, we focus mainly on the converse aspect.

%Further details regarding \cite{Liu2021Rate} and \cite{Yang2024Joint} can be found in the related works subsection of this introduction.

%As mentioned above, in this paper, we consider the nonasymptotic joint data and semantics lossy compression problem defined in our prior work \cite{Yang2024Joint}.
More precisely, in this paper, we consider that data and its semantics are generated by two correlated sources, termed the data and semantic sources, respectively. 
%Note that these two sources are general sources, not necessarily memoryless. 
As semantics are typically inferred from data and cannot be directly observed, we posit that the encoder can only access the data source. Moreover, for practical considerations, we adopt a finite codebook. Therefore, the general JDSLC aims to encode data into a codeword chosen from a finite codebook so that the decoder can recover both the data and its semantics based on the codeword in a lossy sense.
%For a formal problem formulation, please refer to Section II. 
As mentioned in \cite{Yang2024Joint}, our problem formulation adopts a general statistical characterization of the data-semantics relationship, enabling its applicability to a diverse range of semantic communication scenarios \cite{Liu2021Rate, Liu2022Indirect}. Additionally, its nonasymptotic modeling makes it more suitable for practical semantic communication systems with delay and complexity constraints. In this paper, we delve into the converse aspect of this problem, including its nonasymptotic and second-order converse bounds, and illustrate our results using an example of erased fair coin flips.

In the remainder of this section, we provide a brief overview of related works in Subsection \ref{Related_Works_and_Discussions}, followed by an outline of the contributions and organization of this paper in Subsection \ref{Contributions_and_Organization}.

%As mentioned in \cite{Yang2024Joint}, research on our considered problem has a dual motivation. From a semantic communication/compression standpoint, this formulation adopts a general statistical characterization of the data-semantics relationship, enabling its applicability to a diverse range of semantic communication scenarios \cite{Liu2021Rate, Liu2022Indirect}. Additionally, from a nonasymptotic information theory perspective, it is a new problem formulation, and its converse part has yet to be studied.

%defined on product alphabet $\mathcal{X} \times \mathcal{M}$,

%As the semantics of the data are generally not directly observable but can only be inferred from the data, we posit that the compressor, or encoder, has access solely to the data source $X$. 

\subsection{Related Works and Discussions}
\label{Related_Works_and_Discussions}
Shortly after Shannon established his mathematical theory of communication \cite{Shannon1948A}, Weaver identified three levels of communication in \cite{Shannon1949The} as follows:
\begin{itemize}
	\item Level A: How accurately can the symbols of communication be transmitted? (The technical problem.)
	\item Level B: How precisely do the transmitted symbols convey the desired meaning? (The semantic problem.)
	\item Level C: How effectively does the received meaning affect conduct in the desired way? (The effectiveness problem.)
\end{itemize}

%\noindent 

%Although Shannon stated in \cite{Shannon1948A} that ``the semantic aspects of communication are irrelevant to the engineering problem,'' he was actually emphasizing that the semantic aspects do not affect the characterization of the technical problem and are, therefore, not important.
Although Shannon stated in \cite{Shannon1948A} that ``the semantic aspects of communication are irrelevant to the engineering problem,'' he was actually emphasizing that the semantic aspects do not affect the characterization of the technical problem. As clarified by Weaver in \cite{Shannon1949The}, Shannon's statement ``does not mean that the engineering aspects are necessarily irrelevant to the semantic aspects.''

Indeed, Shannon's information theory is general enough and, therefore, of great potential to provide invaluable insights for addressing semantic problems \cite{Shannon1949The, Gunduz2023Beyond}. There have been several works that attempted to leverage Shannon's information theory to analyze semantic communication problems \cite{Liu2021Rate, Liu2022Indirect, Yuxuan2023Rate, Shi2023Excess, Stavrou2023Role}. Specifically, the authors in \cite{Liu2021Rate} and \cite{Liu2022Indirect} proposed a semantic source model comprising an observable source and a correlated hidden source, with each source dedicated to memorylessly generating data and its embedded semantics, respectively. They considered a lossy source coding problem that seeks to simultaneously recover data and semantics under their respective separable distortion constraints by encoding only the observable source. The corresponding semantic rate-distortion function was then obtained. Note that the problem considered in \cite{Liu2021Rate} and \cite{Liu2022Indirect} is an asymptotic version of our problem and inspires our studies. While \cite{Liu2021Rate} and \cite{Liu2022Indirect} focused on the point-to-point source coding scenario, \cite{Yuxuan2023Rate} and \cite{Shi2023Excess} expanded the scope by applying this semantic modeling approach to the multi-terminal source coding and joint source-channel coding scenarios, respectively. \cite{Stavrou2023Role} proposed a Blahut–Arimoto type algorithm to compute the so-called semantic rate-distortion function. \cite{Li2024Fundamental} further gives some analytical properties of the semantic rate-distortion function and introduces a neural network designed for estimating the semantic rate-distortion function from samples. However, the above works focus on first-order asymptotics, in contrast to the nonasymptotic and second-order asymptotics we are investigating.

Shannon's information theory has also been used to help construct practical semantic communication systems. In fact, most recent practical semantic communication systems are constructed based on deep learning (DL) \cite{Luo2022Semantic, Yang2023Semantic}. With the help of the variational approximation technique \cite{Alex2017Deep} and the reparameterization trick \cite{Kingma2014Auto}, information measures, such as entropy and mutual information, can be (approximatively) optimized using stochastic gradient descent. Accordingly, some information-theoretic expressions that represent the performance limits of semantic communication systems can serve as objective functions to guide system design. Some representative works include \cite{Xie2020Deep, Xie2021Deep, Shao2022Learning, Shao2023TaskOriented}.

%Currently, most practical semantic communication systems are constructed based on deep learning (DL) \cite{Luo2022Semantic, Yang2023Semantic}. Results of Shannon's information theory have also been used to help construct DL-based semantic communication systems. More precisely, with the help of the variational approximation technique \cite{Alex2017Deep} and the reparameterization trick \cite{Kingma2014Auto}, information measures, such as entropy and mutual information, can be (approximatively) optimized using stochastic gradient descent. Accordingly, some information-theoretic expressions that represent the performance limits of semantic communication systems can serve as objective functions to guide system design. Some representative works include \cite{Xie2020Deep, Xie2021Deep, Shao2022Learning, Shao2023TaskOriented}, etc.

From another line of research, finite blocklength and second-order analysis of source coding traces back to the seminal work of Strassen \cite{Volker1962Asymptotische} in 1962, in which he investigated the scenario of almost lossless source coding. After decades of being almost forgotten, there has been a significant revival of research on these characterizations, such as in \cite{Ingber2011Dispersion, Kostina2012Fixed, Kostina2013LossyJoint, Kostina2016Nonasymptotic, Zhou2019NonAsymptotic, Yang2024Joint}. Specifically, in \cite{Ingber2011Dispersion}, the authors found the dispersion of lossy source coding for finite alphabet sources using a type-based approach. The case of i.i.d. Gaussian source with quadratic distortion measure was also treated in \cite{Ingber2011Dispersion}. In \cite{Kostina2012Fixed}, the authors proposed several general nonasymptotic achievability and converse bounds for lossy source coding. Moreover, they derived second-order asymptotics valid for sources with abstract alphabets through asymptotic analysis of the nonasymptotic bounds. The analysis methods employed in \cite{Kostina2012Fixed} are subsequently generalized to apply to scenarios such as lossy joint source-channel coding \cite{Kostina2013LossyJoint}, hidden lossy source coding \cite{Kostina2016Nonasymptotic}, Kaspi problem and Fu-Yeung problem \cite{Zhou2019NonAsymptotic}, and our considered JDSLC problem \cite{Yang2024Joint}. In particular, our prior work \cite{Yang2024Joint} proposed nonasymptotic and second-order achievability bounds for our considered general JDSLC problem, i.e., an extended nonasymptotic hidden lossy source coding problem wherein both the hidden and observable sources are to be recovered. However, the converse aspect of this problem has yet to be thoroughly studied. This paper is intended to fill this gap.

%In this paper, we further study the converse aspect of this problem.

\subsection{Contributions and Organization}
\label{Contributions_and_Organization}
The main contributions of this paper are listed as follows:
\begin{itemize}
	%	\item We introduce a nonasymptotic information-theoretic analysis framework to investigate the nonasymptotic fundamental limits of joint data and semantics lossy compression.
%	\item We introduce an information-theoretic nonasymptotic analysis framework to investigate the nonasymptotic fundamental limits of joint data and semantics lossy compression.
%	\item For stationary memoryless sources, separable distortion measures, and appropriate maximum admissible distortions, we present a second-order achievability bound by applying the two-dimensional Berry–Ess\'een theorem to our nonasymptotic bounds.

	\item We derive general nonasymptotic converse bounds for joint data and semantics lossy compression. These bounds are valid for general sources and distortion measures.
	
	\item For stationary memoryless sources, separable distortion measures, and appropriate maximum admissible distortions, we find the tight second-order converse bound, thereby, the dispersion of joint data and semantics lossy compression.
	
	\item We derive the semantic rate-distortion function of the case of erased fair coin flips. Nonasymptotic achievability and converse bounds tailored for this case are also derived.
%	Under the erased fair coin flips case, we derive its rate-distortion function, nonasymptotic achievability and converse bounds.
	
\end{itemize}
Numerical results demonstrate that our second-order asymptotic results effectively approximate the optimum rate at given blocklengths.

%Numerical results demonstrate that our second-order asymptotic approximation results well approximate the optimum rate at given blocklengths.

The remainder of this paper is organized as follows. In Section \ref{SEC_2}, we formally define the joint data and semantics lossy problem, then introduce basic notations and some useful properties. In Sections \ref{SEC_3} and \ref{SEC_4}, we show the general nonasymptotic converse bounds and the second-order asymptotics, respectively. The analysis results under the case of erased fair coin flips are derived and depicted in Section \ref{SEC_5}. Finally, Section \ref{SEC_6} concludes the paper.

%The example of erased fair coin flips is analyzed and visually discussed in Section V. 

%Section III derives the general nonasymptotic converse bounds. Section IV gives the second-order asymptotic results. Section V introduces the case of erased fair coin flips, gives its rate-distortion function, nonasymptotic achievability and converse bounds, and 

%The general nonasymptotic converse bounds and the second-order asymptotic results are derived in Sections III and IV, respectively. Finally, conclusions are drawn in Section V.

%Second-order or dispersion results for information theory problems—in which the behavior of the maximum achievable rate as a function of blocklength is determined with fixed probability of error—date back to Strassen [3]. There has been significant renewed interest of late in such characterizations, particularly since the work of [4].
%
%From another line of research, research on finite blocklength data compression is originated from the seminal work \cite{Volker1962Asymptotische} in 1962, which developed the second-order refinement for almost lossless data compression.

\section{Preliminaries}
\label{SEC_2}

While introducing some new expressions, the notations and definitions in this paper are fully compatible with those in our prior work \cite{Yang2024Joint}. Even so, for the sake of readability and completeness of this paper, we restate them as follows.

%For the readability and completeness of this paper, we repeat as follows.

\subsection{JDSLC in the General and Block Settings}

%\includegraphics[width=0.8\linewidth]{example-image-duck}
%\begin{figure*}[htbp]
%	\centering
%	\includegraphics[width=0.6\linewidth]{Joint_Data_and_Semantics_Compression_Converse_Version_General.eps}
%	\caption{Joint data and semantics lossy compression in the nonasymptotic regime.}
%	\label{Joint_Data_and_Semantics_Compression_Conv_General_fig}
%\end{figure*}
%\begin{figure*}[htbp]
%	\centering
%	\includegraphics[width=0.6\linewidth]{Joint_Data_and_Semantics_Compression_Converse_Version_Block.eps}
%	\caption{Joint data and semantics lossy compression in the nonasymptotic regime.}
%	\label{Joint_Data_and_Semantics_Compression_Conv_Block_fig}
%\end{figure*}

\begin{figure*}
	\centering
	\subfigure[An $(M, d_s, d_x, \epsilon)$ code.]{
		\begin{minipage}{0.6\textwidth}
			\includegraphics[width=\textwidth]{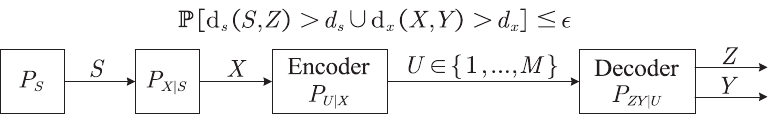} \\
		\end{minipage}}
	\subfigure[A $(k, M, d_s, d_x, \epsilon)$ code.]{
		\begin{minipage}{0.6\textwidth}
			\includegraphics[width=\textwidth]{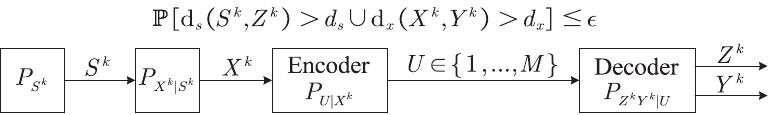} \\
		\end{minipage}}
	\caption{Joint data and semantics lossy compression in the nonasymptotic regime.} 
	\label{JDSLC_figs}
\end{figure*}

We consider that the data and the underlying semantics are drawn from two correlated sources, $S$ and $X$, respectively. Denote the joint distribution of $(S, X)$ as $P_{SX}$, defined on product alphabet $\mathcal{M} \times \mathcal{X}$. Denote the reconstruction alphabets of $S$ and $X$ as $\widehat{\mathcal{M}}$ and $\widehat{\mathcal{X}}$, respectively. We assume finite alphabets $\mathcal{M}$,  $\widehat{\mathcal{M}}$, $\mathcal{X}$, and $\widehat{\mathcal{X}}$, quantifiable data, and semantic distortions. The corresponding distortion measures are denoted as $\textsf{d}_s: \mathcal{M} \times \widehat{\mathcal{M}} \mapsto [0, +\infty)$ and $\textsf{d}_x: \mathcal{X} \times \widehat{\mathcal{X}} \mapsto [0, +\infty)$, respectively. Under the general setting above, we define the $(M, d_s, d_x, \epsilon)$ code as follows.
\begin{definition}
(General):
An $(M, d_s, d_x, \epsilon)$ code 
for $\{\mathcal{M}, \mathcal{X}, \widehat{\mathcal{M}}, \widehat{\mathcal{X}}, P_{SX}, \textsf{d}_s: \mathcal{M} \times \widehat{\mathcal{M}} \mapsto [0, +\infty), \textsf{d}_x: \mathcal{X} \times \widehat{\mathcal{X}} \mapsto [0, +\infty)\}$ 
is a pair of random mappings $P_{U|X}: \mathcal{X} \to \{1, \dots, M\}$ and $P_{ZY|U}: \{1,\dots, M\} \to \widehat{\mathcal{M}}\times \widehat{\mathcal{X}}$ such that the joint excess distortion probability satisfies $\mathbb{P}[\textsf{d}_s(S,Z)>d_s \cup \textsf{d}_x(X,Y)>d_x] \leq \epsilon$.
\end{definition}

%The block setting is a special case of the general setting in which $\mathcal{M}$, $\mathcal{X}$, $\widehat{\mathcal{M}}$, and $\widehat{\mathcal{X}}$ are all $k$-fold Cartesian products of, say, $\mathcal{S}$, $\mathcal{A}$, $\hat{\mathcal{S}}$, and $\hat{\mathcal{A}}$, respectively.

The block setting is a special case of the general setting, induced by the specialization that the alphabets $\mathcal{M}$, $\mathcal{X}$, $\widehat{\mathcal{M}}$, and $\widehat{\mathcal{X}}$ are all $k$-fold Cartesian products. Note that $k$ is referred to as the blocklength. In block setting with $\mathcal{M} = \mathcal{S}^k$, $\mathcal{X}=\mathcal{A}^k$, $\widehat{\mathcal{M}} = \hat{\mathcal{S}}^k$, and $\widehat{\mathcal{X}}=\hat{\mathcal{A}}^k$, the $(k, M, d_s, d_x, \epsilon)$ code is defined as follows.
%of $\mathcal{S}$, $\mathcal{A}$, $\hat{\mathcal{S}}$, and $\hat{\mathcal{A}}$, respectively, the code is define as follows.
\begin{definition}
	(Block):
	A $(k, M, d_s, d_x, \epsilon)$ code is an $(M, d_s, d_x, \epsilon)$ code for $\{\mathcal{S}^k, \mathcal{A}^k, \hat{\mathcal{S}}^k, \hat{\mathcal{A}}^k, P_{S^kX^k}, \textsf{d}_s: \mathcal{S}^k \times \hat{\mathcal{S}}^k \mapsto [0, +\infty), \textsf{d}_x: \mathcal{A}^k \times \hat{\mathcal{A}}^k \mapsto [0, +\infty)\}$ .
%	for $k$-fold Cartesian product alphabets $\mathcal{M} = \mathcal{S}^k$, $\mathcal{X}=\mathcal{A}^k$, $\widehat{\mathcal{M}} = \hat{\mathcal{S}}^k$, and $\widehat{\mathcal{X}}=\hat{\mathcal{A}}^k$.
%	a pair of random mappings $P_{U|X}: \mathcal{X} \to \{1, \dots, M\}$ and $P_{ZY|U}: \{1,\dots, M\} \to \widehat{\mathcal{M}}\times \widehat{\mathcal{X}}$ such that $\mathbb{P}[\textsf{d}_s(S,Z)>d_s \cup \textsf{d}_x(X,Y)>d_x] \leq \epsilon$.
\end{definition}
The nonasymptotic joint data and semantics lossy compression problem in the general and block settings are illustrated in Fig. \ref{JDSLC_figs}.

%\begin{remark} 
%	\label{remark_model_advantages}
%	Since the semantic structure in our model, which is formed by the semantic source $S$, the copula of $(S,X)$, and the semantic distortion measure $\textsf{d}_s$, is allowed to be carefully designed to fit the task/service/content, our model possesses great potential to capture the semantic aspect of source $X$. Meanwhile, assuming the predetermined semantic structure allows us to focus exclusively on the technical aspects of the problem, free from the influence of semantic considerations.
%	%	On the other hand, by assuming that the semantic structure above is given, we can then concentrate on the communication aspects of the problem, free from the influence of semantic considerations.
%\end{remark}

%By assuming that the above semantic structure is given in advance, we are able to focus on the communication aspects of the problem without being influenced by the semantic considerations.	
%%In other words, the meaning is represented by the setting of the intrinsic state and the semantic distortion measure.
%metrics

\subsection{Tilted Information}
The noisy rate-distortion function is defined as \cite{Liu2021Rate}
\begin{subequations}
	\label{rate_distortion_function}
	\begin{align}
	R_{S,X}(d_s,d_x) \triangleq \min_{P_{ZY|X}}\ &I(X;Z,Y) \label{rate_distortion_function_obj} \\
	\mathrm{s.t.}\
	&\mathbb{E}\left[\bar{\textsf{d}}_s(X,Z)\right] \leq d_s, \\
	&\mathbb{E}\left[\textsf{d}_x(X,Y)\right] \leq d_x, 
	\end{align}
\end{subequations}
where $\bar{\textsf{d}}_s: \mathcal{X} \times \widehat{\mathcal{M}} \mapsto [0, +\infty]$ is given by
\begin{equation}
\bar{\textsf{d}}_s(x,z)\triangleq \mathbb{E}[\textsf{d}_s(S, z)| X=x].
\end{equation}
Define $d_{s,\min}\triangleq \mathbb{E}[\min_{z \in \widehat{\mathcal{M}}} \bar{\textsf{d}}_s(X,z)]$ and $d_{x,\min}\triangleq \mathbb{E}[\min_{y \in \widehat{\mathcal{X}}} \textsf{d}_x(X,y)]$. By \cite{Yang2024Joint}, $R_{S,X}(d_s,d_x)$ exists if and only if maximum admissible distortions $(d_s,d_x) \in \mathcal{D}_{\mathrm{adm}} \triangleq \{(d_s,d_x): d_s \geq d_{s,\min}, d_x \geq d_{x,\min}\}$.
%Note that the minimum in \eqref{rate_distortion_function} can be achieved by some $P_{ZY|X}$ only for maximum admissible distortions $(d_s,d_x)$ satisfying $d_s \geq d_{s,\min}\triangleq \mathbb{E}[\min_{z \in \widehat{\mathcal{M}}} \bar{\textsf{d}}_s(X,z)]$ and $d_x \geq d_{x,\min}\triangleq \mathbb{E}[\min_{y \in \widehat{\mathcal{X}}} \textsf{d}_x(X,y)]$. Thus $R_{S,X}(d_s,d_x)$ exists if and only if $(d_s,d_x) \in \mathcal{D}_{\mathrm{adm}} \triangleq \{(d_s,d_x): d_s \geq d_{s,\min}, d_x \geq d_{x,\min}\}$. 
We further define
\begin{align}
R_{\tilde{X}}(d_s,d_x) \triangleq R_{S,\tilde{X}}(d_s,d_x),
\end{align}
where $P_{S\tilde{X}} = P_{\tilde{X}}P_{S|X}$. Clearly, $R_{X}(d_s,d_x) = R_{S,X}(d_s,d_x)$.

Let $\mathcal{P}^{\star}(d_s,d_x)$ denote the set of optimal solutions of problem \eqref{rate_distortion_function} with $(d_s,d_x) \in \mathcal{D}_{\mathrm{adm}}$. Define sets
\begin{align}
\mathcal{D}_{sx} \triangleq \{&(d_s,d_x) \in \mathcal{D}_{\mathrm{adm}}: \forall P_{Z^\star Y^\star|X} \in \mathcal{P}^{\star}(d_s,d_x),\nonumber \\ &\mathbb{E}\left[\bar{\textsf{d}}_s(X,Z^\star)\right] = d_s \textrm{ and }  \mathbb{E}\left[\textsf{d}_x(X,Y^\star)\right] = d_x\},
\end{align}
\begin{align}
\mathcal{D}_{\bar{s}x} \triangleq \{&(d_s,d_x)\in \mathcal{D}_{\mathrm{adm}}: \exists P_{Z^\star Y^\star|X} \in \mathcal{P}^{\star}(d_s,d_x) \textrm{ such that} \nonumber \\ &\mathbb{E}\left[\bar{\textsf{d}}_s(X,Z^\star)\right] < d_s;\ \forall P_{Z^\star Y^\star|X} \in \mathcal{P}^{\star}(d_s,d_x), \nonumber \\ &\mathbb{E}\left[\textsf{d}_x(X,Y^\star)\right] = d_x\},
\end{align}
\begin{align}
\mathcal{D}_{s\bar{x}} \triangleq \{&(d_s,d_x)\in \mathcal{D}_{\mathrm{adm}}: \forall P_{Z^\star Y^\star|X} \in \mathcal{P}^{\star}(d_s,d_x), \nonumber \\ &\mathbb{E}\left[\bar{\textsf{d}}_s(X,Z^\star)\right] = d_s;\ \exists P_{Z^\star Y^\star|X} \in \mathcal{P}^{\star}(d_s,d_x) \nonumber \\ &\textrm{ such that } \mathbb{E}\left[\textsf{d}_x(X,Y^\star)\right] < d_x\},
\end{align}
\begin{align}
\mathcal{D}_{\bar{s}\bar{x}} \triangleq \{&(d_s,d_x)\in \mathcal{D}_{\mathrm{adm}}: \exists P_{Z^\star Y^\star|X} \in \mathcal{P}^{\star}(d_s,d_x) \textrm{ such that} \nonumber \\ &\mathbb{E}\left[\bar{\textsf{d}}_s(X,Z^\star)\right] < d_s;\ \exists P_{Z^\star Y^\star|X} \in \mathcal{P}^{\star}(d_s,d_x) \nonumber \\ &\textrm{ such that } \mathbb{E}\left[\textsf{d}_x(X,Y^\star)\right] < d_x\}.
\end{align}
Clearly, $\{\mathcal{D}_{sx}, \mathcal{D}_{\bar{s}x}, \mathcal{D}_{s\bar{x}}, \mathcal{D}_{\bar{s}\bar{x}}\}$ forms a partition of $\mathcal{D}_{\mathrm{adm}}$. As clarified in \cite{Yang2024Joint}, $R_{S,X}(d_s,d_x)$ is differentiable on $\mathcal{D}_{\mathrm{in}} \triangleq \textsf{int}(\mathcal{D}_{sx}) \cup \textsf{int}(\mathcal{D}_{\bar{s}x}) \cup \textsf{int}(\mathcal{D}_{s\bar{x}}) \cup \textsf{int}(\mathcal{D}_{\bar{s}\bar{x}})$, where $\textsf{int}(\cdot)$ denotes the interior of the input set.

%Denote by $P_{Z^\star Y^\star|X}$ (one of) the optimal solution of problem \eqref{rate_distortion_function} under some given $(d_s,d_x)$ such that $R_{S,X}(d_s,d_x) < \infty$.

%We make the following basic assumptions:
%\begin{assumption}
%	\label{Assumptions}
%	Assume that
%	\begin{itemize}
%		\item [(a)] \bluemark{Set $\mathcal{D}_e \triangleq \{(d_s,d_x): \ \textrm{both}\ \partial R_{S,X}(d_s,d_x)/\partial d_x$ and $\partial R_{S,X}(d_s,d_x)/\partial d_s\ \textrm{exist}\}$ is nonempty;} \label{assumption1}
%		\item [(b)] \bluemark{The infimum in \eqref{rate_distortion_function} is achieved by some $P_{Z^\star Y^\star |X}$ such that $\mathbb{E}\left[\bar{\textsf{d}}_s(X,Z^\star)\right] = d_s$ or $\mathbb{E}\left[\textsf{d}_x(X,Y^\star)\right] = d_x$.} \label{assumption2}
%	\end{itemize}
%\end{assumption}

For $(d_s,d_x) \in \mathcal{D}_{\mathrm{in}}$, the noisy $(\textsf{d}_s,\textsf{d}_x)$-tilted information in $(s,x) \in \mathcal{M} \times \mathcal{X}$ given representations $z \in \widehat{\mathcal{M}}$ and $y \in \widehat{\mathcal{X}}$ is defined as
\begin{equation}
\label{tilted_information}
\begin{aligned}
\tilde{\jmath}_{S,X}(s,x,z,y,d_s,d_x) \triangleq& \imath_{X;Z^\star Y^\star}(x;z,y) + \lambda_s^\star \textsf{d}_s(s,z)\\ 
&+ \lambda_x^\star \textsf{d}_x(x,y) - \lambda_s^\star d_s - \lambda_x^\star d_x,
\end{aligned}
\end{equation}
where $P_{Z^\star Y^\star| X} \in \mathcal{P}^{\star}(d_s,d_x)$, $P_{Z^\star Y^\star}=\sum_{x \in \mathcal{X}}P_XP_{Z^\star Y^\star| X}$, and
\begin{equation}
\label{information_density}
\imath_{X;Z^\star Y^\star}(x;z,y) \triangleq \log \frac{\mathrm{d} P_{Z^\star Y^\star| X=x}}{\mathrm{d} P_{Z^\star Y^\star}}(z,y),
\end{equation}
\begin{equation}
\label{lambda_s}
\lambda_s^\star \triangleq -\frac{\partial R_{S,X}(d_s,d_x)}{\partial d_s},
%	\footnote{$\frac{\partial R_{S,X}(d_s,d_x)}{\partial d_s}$ denotes $\frac{\partial R_{S,X}(a,b)}{\partial a}\Big|_{a=d_s,b=d_x}$.}
\end{equation}
\begin{equation}
\label{lambda_x}
\lambda_x^\star \triangleq -\frac{\partial R_{S,X}(d_s,d_x)}{\partial d_x}.
\end{equation}
For $(d_s,d_x) \in \mathcal{D}_{\mathrm{in}}$, the $(\bar{\textsf{d}}_s,\textsf{d}_x)$-tilted information in $x$ for the surrogate noiseless two-constraint source coding problem \cite[Section VI]{Blahut1972Computation}, \cite[Problem 10.19]{Cover2006Elements} is defined as
\begin{equation}
\begin{aligned}
&\jmath_{X}(x,d_s,d_x)\\ 
\triangleq &\log\! \frac{1}{\mathbb{E}[\exp\{\lambda_s^\star d_s \! + \! \lambda_x^\star d_x \!-\! \lambda_s^\star \bar{\textsf{d}}_s(x,Z^\star) \!-\! \lambda_x^\star \textsf{d}_x(x,Y^\star)\}]},
\end{aligned}
\end{equation}
where the expectation is with respect to the unconditional distribution $P_{Z^\star Y^\star}$. By the differentiability of $R_{S,X}(d_s,d_x)$ on $\mathcal{D}_{\mathrm{in}}$, the tilted informations $\tilde{\jmath}_{S,X}(s,x,z,y,d_s,d_x)$ and $\jmath_{X}(x,d_s,d_x)$ are well-defined. The following properties of $\jmath_{X}(x,d_s,d_x)$ have been proved in \cite{Yang2024Joint}.
\begin{property} 
	\label{tilted_information_property}
	(\cite[Property 1]{Yang2024Joint}): Fix $(d_s,d_x) \in \mathcal{D}_{\mathrm{in}}$. For $P_{Z^\star Y^\star}$-a.e. $(z,y)$, it holds that
	\begin{equation}
	\label{d_tilted_information_of_surrogate}
	\begin{aligned}
	\jmath_{X}(x,d_s,d_x) = &\imath_{X;Z^\star Y^\star}(x;z,y) + \lambda_s^\star \bar{\textsf{d}}_s(x,z)\\ 
	&+ \lambda_x^\star \textsf{d}_x(x,y) - \lambda_s^\star d_s - \lambda_x^\star d_x,
	\end{aligned}
	\end{equation}
	where $P_{XZ^\star Y^\star} = P_X P_{Z^\star Y^\star | X}$. Moreover,
	\begin{align}
	R_{S,X}(d_s,d_x) =& \min_{P_{ZY|X}} \mathbb{E}\big[\imath_{X;ZY}(X;Z,Y) + \lambda_s^\star \bar{\textsf{d}}_s(X,Z) \nonumber \\ 
	&\quad \quad \ \, + \lambda_x^\star \textsf{d}_x(X,Y)\big] - \lambda_s^\star d_s - \lambda_x^\star d_x \label{property1} \\
	%	=& \min_{P_{ZY|X}} \mathbb{E}\big[\imath_{X;Z^\star Y^\star}(X;Z,Y) + \lambda_s^\star \bar{\textsf{d}}_s(X,Z) \nonumber \\ 
	%	&+ \lambda_x^\star \textsf{d}_x(X,Y)\big] - \lambda_s^\star d_s - \lambda_x^\star d_x \label{property2} \\
	=& \mathbb{E}\big[\jmath_{X}(X,d_s,d_x)\big], \label{property3}
	\end{align}
	and for all $z \in \widehat{\mathcal{M}}$ and $y \in \widehat{\mathcal{X}}$,
	\begin{equation}
	\label{property4}
	\begin{aligned}
	&\mathbb{E}\big[\exp\big\{\lambda_s^\star d_s + \lambda_x^\star d_x - \lambda_s^\star \bar{\textsf{d}}_s(X,z) - \lambda_x^\star \textsf{d}_x(X,y)\\ 
	&+ \jmath_{X}(X,d_s,d_x) \big\}\big] \leq 1
	\end{aligned}
	\end{equation}
	with equality for $P_{Z^\star Y^\star}$-a.e. $(z,y)$.
\end{property}
%\begin{proof}[Proof]
%	See Appendix \ref{proof_tilted_information_property}.
%\end{proof}
%\begin{remark} 
%	\label{remark_core_properties}
%	The optimal solutions to problem \eqref{rate_distortion_function}, which in turn achieve $R_{S,X}(d_s,d_x)$, consistently adhere to Property \ref{tilted_information_property}. Nevertheless, the conditional distributions $P_{Z^\star Y^\star|X}$ that satisfy Property \ref{tilted_information_property} are not necessarily feasible to problem \eqref{rate_distortion_function}.
%\end{remark}

By Property \ref{tilted_information_property}, for $P_{Z^\star Y^\star}$-a.e. $(z,y)$,
\begin{equation}
\label{relation_of_tilted_information}
\tilde{\jmath}_{S,X}\!(s,x,z,y,d_s,d_x) \!=\! \jmath_{X}\!(x,d_s,d_x) + \lambda_s^\star \textsf{d}_s(s,z) - \lambda_s^\star \bar{\textsf{d}}_s(x,z).
\end{equation}
Juxtaposing with \eqref{tilted_information}, we have
\begin{align}
R_{S,X}(d_s,d_x) =& \mathbb{E}[\tilde{\jmath}_{S,X}(S,X,Z^\star,Y^\star,d_s,d_x)] \nonumber \\ 
=& \mathbb{E}[\jmath_{X}(X,d_s,d_x)].
\end{align}
Define the noisy rate-dispersion function as
\begin{equation}
\label{noisy_rate_dispersion_func}
\tilde{\mathcal{V}}(d_s,d_x) \triangleq \textrm{Var}\left[\tilde{\jmath}_{S,X}(S,X,Z^\star,Y^\star,d_s,d_x)\right].
\end{equation}
Similarly, define the rate-dispersion function of the surrogate noiseless problem as
\begin{equation}
\mathcal{V}(d_s,d_x) \triangleq \textrm{Var}\left[\jmath_{X}(X,d_s,d_x)\right].
\end{equation}
The following proposition reveals the relationship between $\tilde{\mathcal{V}}(d_s,d_x)$ and $\mathcal{V}(d_s,d_x)$.
\begin{proposition} \label{Prop_relationship_V_tildeV}
	(\cite[Proposition 1]{Yang2024Joint}):
	$\tilde{\mathcal{V}}(d_s,d_x)$ can be written as
	\begin{equation}
	\label{relationship_V_tildeV}
	\tilde{\mathcal{V}}(d_s,d_x) = \mathcal{V}(d_s,d_x) + \lambda_s^{\star 2}\textrm{Var}\left[\textsf{d}_s(S,Z^\star)|X,Z^\star\right],
	\end{equation}
	where $\textrm{Var}\left[U|V\right] \triangleq \mathbb{E}\left[(U - \mathbb{E}\left[U|V\right])^2\right]$.
\end{proposition}
%\begin{proof}[Proof]
%	%	Similar to the proof of \cite[Proposition 1]{Kostina2016Nonasymptotic},
%	\begin{align}
%	\tilde{\mathcal{V}}(d_s,d_x) =& \textrm{Var}\big[\jmath_{X}(X,d_s,d_x) + \lambda_s^\star \textsf{d}_s(S,Z^\star) \nonumber \\
%	&- \lambda_s^\star \mathbb{E}[\textsf{d}_s(S,Z^\star)|X,Z^\star]\big] \\
%	=& \mathcal{V}(d_s,d_x) + \lambda_s^{\star 2}\textrm{Var}\left[\textsf{d}_s(S,Z^\star)|X,Z^\star\right]\nonumber\\ 
%	&+ 2\lambda_s^{\star} \textrm{Cov}\big(\jmath_{X}(X,d_s,d_x), \textsf{d}_s(S,Z^\star)\nonumber\\
%	&- \mathbb{E}[\textsf{d}_s(S,Z^\star)|X,Z^\star]\big)\\
%	=& \mathcal{V}(d_s,d_x) + \lambda_s^{\star 2}\textrm{Var}\left[\textsf{d}_s(S,Z^\star)|X,Z^\star\right].
%	\end{align}
%\end{proof}

\section{Nonasymptotic Converse Bounds}
\label{SEC_3}
In this section, we provide the general nonasymptotic converse bounds.
%\cite[Theorem 2]{Kostina2016Nonasymptotic}
For fixed $P_X$ and auxiliary conditional distribution $P_{\bar{X}|\bar{Z}\bar{Y}}$, denote
\begin{align}
	f_{\bar{X}|\bar{Z}\bar{Y}}(s,x,z,y) \triangleq& \imath_{\bar{X}|\bar{Z}\bar{Y}\Vert X}(x;z,y) + \sup_{\lambda_s \geq 0} \lambda_s(\textsf{d}_s(s,z) - d_s) \nonumber \\ 
	&+\sup_{\lambda_x \geq 0} \lambda_x(\textsf{d}_x(x,y) - d_x) \!-\! \log M,
\end{align}
where
\begin{equation}
	\imath_{\bar{X}|\bar{Z}\bar{Y}\Vert X}(x;z,y) \triangleq \log \frac{\mathrm{d} P_{\bar{X}|\bar{Z}=z,\bar{Y}=y}}{\mathrm{d} P_X}(x).
\end{equation}
The nonasymptotic converse results can now be stated as follows.
\begin{theorem} 
	\label{general_converse_theorem}
	(Converse): If an $(M, d_s, d_x, \epsilon)$ code exists, then the following inequality must be hold:
	\begin{align}
	\label{general_converse}
		\epsilon \geq& \inf_{\substack{P_{ZY|X}:\\ \mathcal{X} \to \widehat{\mathcal{M}}\times \widehat{\mathcal{X}}}} \sup_{\substack{P_{\bar{X}|\bar{Z}\bar{Y}}:\\ \widehat{\mathcal{M}}\times \widehat{\mathcal{X}} \to \mathcal{X}}} \sup_{\gamma \geq 0}\big\{\mathbb{P}\left[f_{\bar{X}|\bar{Z}\bar{Y}}(S,X,Z,Y) \geq \gamma\right]\nonumber\\ 
		&- \exp(-\gamma)\big\},
	\end{align}
	where the middle supremum is over those $P_{\bar{X}|\bar{Z}\bar{Y}}$ such that Radon-Nikodym derivative of $P_{\bar{X}|\bar{Z}=z,\bar{Y}=y}$ with respect to $P_X$ at $x$ exists for $P_{ZY|X}P_X$-a.e. $(z,y,x)$.
\end{theorem}
\begin{proof}[Proof]
	Following the proof of \cite[Theorem 2]{Kostina2016Nonasymptotic}, let the encoder and the decoder be the random mappings $P_{U|X}$ and $P_{ZY|U}$, respectively, where $U$ takes values in $\{1, \dots, M\}$. Then, for any $\gamma \geq 0$,
	\begin{align}
	&\mathbb{P}\left[f_{\bar{X}|\bar{Z}\bar{Y}}(S,X,Z,Y) \geq \gamma\right]\\ 
	=& \mathbb{P}\left[f_{\bar{X}|\bar{Z}\bar{Y}}(S,X,Z,Y) \!\geq\! \gamma, \textsf{d}_s(S,Z)\! > \! d_s  \cup  \textsf{d}_x(X,Y) \!>\! d_x \right] \nonumber \\
	& + \mathbb{P}\left[f_{\bar{X}|\bar{Z}\bar{Y}}(S,X,Z,Y) \!\geq\! \gamma, \textsf{d}_s(S,Z) \!\leq\! d_s, \textsf{d}_x(X,Y) \!\leq\! d_x\right]\\
	=& \epsilon + \mathbb{P}\big[\imath_{\bar{X}|\bar{Z}\bar{Y}\Vert X}(X;Z,Y) \geq \gamma +\log M \nonumber \\
	&, \textsf{d}_s(S,Z) \leq d_s, \textsf{d}_x(X,Y) \leq d_x\big] \\ 
	\leq& \epsilon + \mathbb{P}\left[\imath_{\bar{X}|\bar{Z}\bar{Y}\Vert X}(X;Z,Y) \geq \gamma +\log M\right]\\
	\leq& \epsilon + \frac{\exp(-\gamma)}{M}\mathbb{E}\left[\exp(\imath_{\bar{X}|\bar{Z}\bar{Y}\Vert X}(X;Z,Y))\right] \label{Markov_Inequality} \\
	\leq& \epsilon + \frac{\exp(-\gamma)}{M}\sum_{u=1}^M \bigg( \int_{z \in \widehat{\mathcal{M}}, y \in \widehat{\mathcal{X}}} \mathrm{d} P_{ZY|U}(z,y|u)\nonumber \\ 
	&\cdot\int_{x \in \mathcal{X}} \mathrm{d} P_{\bar{X}|\bar{Z}\bar{Y}}(x|z,y)\bigg) \label{Markov_chain_plus_less_than_1} \\
	=& \epsilon + \exp(-\gamma).
	\end{align}
	where \eqref{Markov_Inequality} is by  Markov's inequality, and \eqref{Markov_chain_plus_less_than_1} holds since $X-U-(Z,Y)$ forms a Markov chain in this order and $P_{U|X}(u|x)\leq 1$ for all $(x,u) \in \mathcal{M}\times \{1,\dots, M\}$.
\end{proof}
The following two corollaries of Theorem \ref{general_converse_theorem} form the basis of our second-order analysis.
\begin{corollary} 
	\label{converse_bound_corollary}
	(Converse):
	Any $(M, d_s, d_x, \epsilon)$ code must satisfy
	\begin{align}
	\label{converse_bound2}
	\epsilon \geq& \sup_{\substack{P_{\bar{X}|\bar{Z}\bar{Y}}:\\ \widehat{\mathcal{M}}\times \widehat{\mathcal{X}} \to \mathcal{X}}} \sup_{\gamma \geq 0} \Big\{\mathbb{E}\Big[\inf_{z \in \widehat{\mathcal{M}}, y \in \widehat{\mathcal{X}}}\mathbb{P}\left[f_{\bar{X}|\bar{Z}\bar{Y}}(S,X,z,y) \geq \gamma | X\right]\Big] \nonumber \\ 
	&- \exp(-\gamma)\Big\},
	\end{align}
	where the first supremum is over those $P_{\bar{X}|\bar{Z}\bar{Y}}$ such that Radon-Nikodym derivative of $P_{\bar{X}|\bar{Z}=z,\bar{Y}=y}$ with respect to $P_X$ at $x$ exists for every $z \in \widehat{\mathcal{M}}$, $y \in \widehat{\mathcal{X}}$ and $P_X$-a.e. $x$.
\end{corollary}
\begin{proof}[Proof]
	We weaken \eqref{general_converse} using the max–min inequality and
	\begin{align}
	&\inf_{P_{ZY|X}}\mathbb{P}\left[f_{\bar{X}|\bar{Z}\bar{Y}}(S,X,Z,Y) \geq \gamma\right]\\
	=& \inf_{P_{ZY|X}}\mathbb{E}\left[\mathbb{P}\left[f_{\bar{X}|\bar{Z}\bar{Y}}(S,X,Z,Y) \geq \gamma | X\right]\right]\\
	=&\mathbb{E}\Big[\inf_{z \in \widehat{\mathcal{M}}, y \in \widehat{\mathcal{X}}}\mathbb{P}\left[f_{\bar{X}|\bar{Z}\bar{Y}}(S,X,z,y) \geq \gamma | X\right]\Big].
	\end{align}
\end{proof}

\begin{corollary} 
	\label{converse_bound_corollary2}
	(Converse):
	If $\lambda_s^\star = 0$, any $(M, d_s, d_x, \epsilon)$ code must satisfy
	\begin{equation}
	\label{converse_bound3}
	\epsilon \geq \sup_{\gamma \geq 0} \left\{\mathbb{P}\left[\jmath_{X}(X,d_s,d_x) - \log M \geq \gamma\right] - \exp(-\gamma)\right\}.
	\end{equation}
\end{corollary}
\begin{proof}[Proof]
By setting $P_{\bar{X}|\bar{Z}\bar{Y}} = P_{X|Z^\star Y^\star}$, $\lambda_s = \lambda_s^\star$ and $\lambda_x = \lambda_x^\star$ in \eqref{general_converse}, any $(M, d_s, d_x, \epsilon)$ code must satisfy
\begin{align}
\epsilon \geq &\inf_{\substack{P_{ZY|X}:\\ \mathcal{X} \to \widehat{\mathcal{M}}\times \widehat{\mathcal{X}}}} \sup_{\gamma \geq 0}\big\{\mathbb{P}\big[\imath_{X; Z^\star Y^\star}(x;z,y) + \lambda_s^\star(\textsf{d}_s(s,z) - d_s)\nonumber\\
&+ \lambda_x^\star (\textsf{d}_x(x,y) - d_x) - \log M \geq \gamma\big]
- \exp(-\gamma)\big\}. \label{converse_C2_1}
\end{align}
Since $\lambda_s^\star = 0$, we can replace the term $\lambda_s^\star(\textsf{d}_s(s,z) - d_s)$ in \eqref{converse_C2_1} by $\lambda_s^\star(\bar{\textsf{d}}_s(x,z) - d_s)$. Combining this with \eqref{d_tilted_information_of_surrogate}, we have
\begin{align}
\epsilon \geq& \inf_{\substack{P_{ZY|X}:\\ \mathcal{X} \to \widehat{\mathcal{M}}\times \widehat{\mathcal{X}}}}\!\! \sup_{\gamma \geq 0} \left\{\mathbb{P}\left[\jmath_{X}(X,d_s,d_x) \!-\! \log M \geq \gamma\right] \!-\! \exp(-\gamma)\right\}\\
=& \sup_{\gamma \geq 0} \left\{\mathbb{P}\left[\jmath_{X}(X,d_s,d_x) - \log M \geq \gamma\right] - \exp(-\gamma)\right\}.
\end{align}
%where \eqref{converse_C2_1} is by setting $P_{\bar{X}|\bar{Z}\bar{Y}} = P_{X|Z^\star Y^\star}$, $\lambda_s = \lambda_s^\star$ and $\lambda_x = \lambda_x^\star$ in \eqref{general_converse},
\end{proof}

\section{Asymptotic Analysis}
\label{SEC_4}

In this section, we study second-order asymptotics in blocklength $k$ under the block setting. Recall that in this setting, $\mathcal{M} = \mathcal{S}^k$, $\mathcal{X}=\mathcal{A}^k$, $\widehat{\mathcal{M}} = \hat{\mathcal{S}}^k$, and $\widehat{\mathcal{X}}=\hat{\mathcal{A}}^k$. We make the following assumptions.
%In this section, we consider the alphabets can be written as $\mathcal{M} = \mathcal{S}^k$, $\mathcal{X}=\mathcal{A}^k$, $\widehat{\mathcal{M}} = \hat{\mathcal{S}}^k$, $\widehat{\mathcal{X}}=\hat{\mathcal{A}}^k$, and assume that
\begin{enumerate}
	\item 
	\label{memoryless_sources}
	(Stationary Memoryless Sources):
	$P_{S^k X^k} = P_S P_{X|S} \times \dots \times P_S P_{X|S}$. 
	
	\item 
	\label{separable_distortion}
	(Separable Distortion Measures):
	\begin{equation}
	\textsf{d}_s(s^k,z^k) = \frac{1}{k}\sum_{i=1}^k \textsf{d}_s(s_i,z_i),
	\end{equation}
	\begin{equation}
	\textsf{d}_x(x^k,y^k) = \frac{1}{k}\sum_{i=1}^k \textsf{d}_x(x_i,y_i).
	\end{equation}

	\item
	\label{Finite_alphabets_for_asymptotic}
	(Finite Alphabets):
	The alphabets $\mathcal{S}$, $\mathcal{A}$, $\hat{\mathcal{S}}$, $\hat{\mathcal{A}}$ are finite sets. 
	
	\item
	\label{differentiability}
	(Differentiability): 
	For all $P_{\bar{X}}$ in some neighborhood of $P_X$, $\textsf{supp}(P_{\bar{Z}^\star \bar{Y}^\star}) = \textsf{supp}(P_{Z^\star Y^\star})$, where $P_{\bar{Z}^\star \bar{Y}^\star}$ achieves $R_{\bar{X}}(d_s,d_x)$; $R_{\bar{X}}(d_s,d_x)$ is twice continuously differentiable with respect to $P_{\bar{X}}$.
	
	\item
	\label{Assumption_positive_definite}
	(Non-Degenerate Single-Letter Distortion Measures):
	There exist possibly identical elements $x_1,x_2 \in \mathcal{A}$ such that functions $\bar{\textsf{d}}_s(x_1,\cdot)$ and $\textsf{d}_x(x_2,\cdot)$ vary non-constantly across the elements in $\hat{\mathcal{S}}$ and $\hat{\mathcal{A}}$, respectively. 
\end{enumerate}
The following theorem is derived through an asymptotic analysis of Corollarys \ref{converse_bound_corollary} and \ref{converse_bound_corollary2}.
\begin{theorem} 
	\label{theorem_Gaussian_approximation}
	(Converse, Second-Order Asymptotics): Under assumptions \ref{memoryless_sources}-\ref{Assumption_positive_definite}, fixing $(d_s, d_x) \in \mathcal{D}_{\mathrm{in}}$ and $0 < \epsilon < 1$, any codebook size $M$ compatible with excess distortion constraint $\mathbb{P}[\textsf{d}_s(S^k,Z^k)>d_s \cup \textsf{d}_x(X^k,Y^k)>d_x] \leq \epsilon$ must satisfy
	\begin{align}
	\log M \geq k R_{S,X}(d_s,d_x)+ \sqrt{k \tilde{\mathcal{V}}(d_s,d_x)} Q^{-1}(\epsilon) + O(\log k),
	\end{align}
	where $Q^{-1}(\cdot)$ denotes the inverse of the complementary standard Gaussian cumulative distribution function, and $f(k)=O(g(k))$ means $\lim\sup_{k \to \infty}\left|f(k)/g(k)\right| < \infty$.
\end{theorem}
\begin{proof}[Proof]
	See Appendix \ref{proof_converse_theorem_Gaussian_approximation}.
\end{proof}

Define the minimum achievable codebook size at blocklength $k$, maximum admissible distortions $(d_s,d_x)$, and joint excess distortion probability $\epsilon$ as
\begin{equation}
    M^\star(k, d_s, d_x, \epsilon) \triangleq \min \{M: \exists (k, M, d_s, d_x, \epsilon)\ \textrm{code}\}.
\end{equation}
Then, we have the following dispersion theorem for joint data and semantics lossy compression.
\begin{theorem} 
	\label{theorem_dispersion}
	(Second-Order Asymptotics): Under assumptions \ref{memoryless_sources}-\ref{Assumption_positive_definite}, fixing $(d_s, d_x) \in \textsf{int}(\mathcal{D}_{sx})$ and $0 < \epsilon < 1$, the minimum achievable codebook size $M^\star(k, d_s, d_x, \epsilon)$ satisfies
	\begin{align}
	\log M^\star(k, d_s, d_x, \epsilon) =& k R_{S,X}(d_s,d_x)+ \sqrt{k \tilde{\mathcal{V}}(d_s,d_x)} Q^{-1}(\epsilon) \nonumber \\ 
	&+ O(\log k).
	\end{align}
\end{theorem}
\begin{proof}[Proof]
	 \cite[Theorem 3]{Yang2024Joint} provides the achievability part of this theorem, and the converse part is given in Theorem \ref{theorem_Gaussian_approximation} in this paper.
\end{proof}

Theorem \ref{theorem_dispersion} provides a closed-form second-order approximation of the optimum finite blocklength coding rate $\log M^\star(k, d_s, d_x, \epsilon) / k$, i.e., for $(d_s, d_x) \in \textsf{int}(\mathcal{D}_{sx})$,
\begin{align}
\label{second_order_cong}
\frac{\log M^\star(k, d_s, d_x, \epsilon)}{k} \cong R_{S,X}(d_s,d_x)+ \sqrt{\frac{\tilde{\mathcal{V}}(d_s,d_x)}{k}} Q^{-1}(\epsilon),
\end{align}
where the notation $\cong$ denotes that the equality holds up to a term of $O\left(\log k/k\right)$. In the next section, we illustrate the approximation accuracy of \eqref{second_order_cong} in the case of erased fair coin flips.

\section{Case Study: Erased Fair Coin Flips}
\label{SEC_5}

In the erased fair coin flips (EFCF) case, we examine two correlated sources, $S$ and $X$. Source $S$ is a binary equiprobable source taking values in $\{0,1\}$, and source $X$, taking values in $\{0,1,e\}$, corresponds to the source obtained by observing source $S$ through a binary erasure channel with an erasure rate of $\delta$. We consider Hamming distortion measures, i.e., $\textsf{d}_s(s^k,z^k) = \frac{1}{k}\sum_{i=1}^k \textrm{1}\{s_i \neq z_i\}$ and $\textsf{d}_x(x^k,y^k) = \frac{1}{k}\sum_{i=1}^k \textrm{1}\{x_i \neq y_i\}$.

\subsection{Rate-Distortion Function}

In this section, we abbreviate $R_{S,X}(d_s,d_x)$ as $R(d_s,d_x)$.
For $0 < \delta \leq 1/3$, $d_s \geq \delta/2$ and $d_x \geq 0$, the rate-distortion function $R(d_s,d_x)$ is given by the following theorem.
\begin{theorem} 
	\label{theorem_Rate_Distortion_func}
	Let $0 < \delta < 1/3$, $d_s \geq \delta / 2$, and $d_x \geq 0$. Define sets
	\begin{itemize}
		\item [(i)] $\mathcal{D}_1 \triangleq \{(d_s,d_x): 0 \leq d_x \leq 2 \delta,\ d_s \geq d_x/2 + \delta/2 \}$;
		\item [(ii)] $\mathcal{D}_2 \triangleq \{(d_s,d_x): 2 \delta \leq d_x \leq 1/2 + \delta/2,\ d_s \geq d_x - \delta/2 \}$;
		\item [(iii)] $\mathcal{D}_3 \triangleq \{(d_s,d_x): d_x \geq d_s + \delta/2,\  \delta/2 \leq d_s \leq 1/2 \}$;
		\item [(iv)] $\mathcal{D}_4 \triangleq \{(d_s,d_x): 2d_s - \delta \leq d_x \leq d_s + \delta/2,\  \delta/2 \leq d_s \leq 3\delta/2 \}$;
		\item [(v)] $\mathcal{D}_5 \triangleq \{(d_s,d_x): d_x \geq 1/2 + \delta/2, \ d_s \geq 1 / 2 \}$.
	\end{itemize}
	Then,
	\begin{itemize}
		\item [(i)] if $(d_s, d_x) \in \mathcal{D}_1$, 
		\begin{equation}
			R(d_s,d_x) = h(\delta) + (1 - \delta) \log 2 - h(d_x) - d_x \log 2,
		\end{equation}
		which is achieved by $X-Y^{\star}-Z^{\star}$ with
		\begin{equation}
			P_{Y^\star}(0) \!=\! P_{Y^\star}(1) \!=\! \frac{1 \!-\! \delta \!-\! d_x}{2 \!-\! 3 d_x}, P_{Y^\star}(e) \!=\! \frac{2\delta \!-\! d_x}{2 \!-\! 3 d_x},
		\end{equation}
		\begin{equation}
			P_{X|Y^{\star}}(x|y)=\left\{
			\begin{array}{rcl}
			1 - d_x & & {x=y}\\
			d_x/2 & & {\mbox{otherwise}}
			\end{array} \right.,
		\end{equation}
		and
		\begin{equation}
			P_{Z^{\star}|Y^{\star}}(z|y)=\left\{
			\begin{array}{rcl}
			1 & & {z=y}\\
			1/2 & & {y = e}\\
			0 & & {\mbox{otherwise}}
			\end{array} \right.;
		\end{equation}
%		\begin{equation}
%		P_{Z^{\star}|Y^{\star}}(z|y)=\left\{
%		\begin{array}{rcl}
%		1 & & {z=y \neq e}\\
%		1/2 & & {y = e}\\
%		0 & & {\mbox{otherwise}}
%		\end{array} \right.
%		\end{equation}
		
		\item [(ii)] if $(d_s, d_x) \in \mathcal{D}_2$, 
		\begin{equation}
			R(d_s,d_x) = (1 - \delta) \left[\log 2 - h((d_x - \delta)/(1 - \delta))\right],
		\end{equation}
		which is achieved by $X-Y^{\star}-Z^{\star}$ with
		\begin{equation}
			P_{Y^\star}(0) = P_{Y^\star}(1) = \frac{1}{2},\ P_{Y^\star}(e) = 0,
		\end{equation}
		\begin{equation}
		P_{X|Y^{\star}}(x|y)=\left\{
		\begin{array}{rcl}
		1 - d_x & & {x=y \neq e}\\
		d_x - \delta & & {x\neq y,\ x, y \neq e}\\
		\delta & & {x= e,\ y \neq e}
		\end{array} \right.,
		\end{equation}
		and
		\begin{equation}
		P_{Z^{\star}|Y^{\star}}(z|y)=\left\{
		\begin{array}{rcl}
		1 & & {z=y \neq e}\\
		0 & & {z \neq y,\ y \neq e}
		\end{array} \right.;
		\end{equation}

		\item [(iii)] if $(d_s, d_x) \in \mathcal{D}_3$,
		\begin{equation}
			R(d_s,d_x) \!=\! (1 \!-\! \delta) \left[\log 2 \!-\! h((d_s \!-\! \delta/2)/(1 \!-\! \delta))\right],
		\end{equation}
		which is achieved by $X - Z^{\star} - Y^{\star}$ with
		\begin{equation}
			P_{Z^\star}(0) = P_{Z^\star}(1) = \frac{1}{2},
		\end{equation}
		\begin{equation}
			P_{X|Z^{\star}}(x|z)\!=\!\left\{
			\begin{array}{rcl}\!\!
			1 \!-\! d_s \!-\! \delta/2 & & {x=z}\\
			d_s \!-\! \delta/2 & & {x\neq z,\ x \neq e}\\
			\delta & & {x= e}
			\end{array} \right.,
			\end{equation}
			and
			\begin{equation}
			P_{Y^{\star}|Z^{\star}}(y|z)=\left\{
			\begin{array}{rcl}
			1 & & {y=z}\\
			0 & & {y \neq z}
			\end{array} \right.;
		\end{equation}
		
		\item [(iv)] if $(d_s, d_x) \in \mathcal{D}_4$, 
		\begin{equation}
		\label{R_D_for_D4}
		\begin{aligned}
		R(d_s,d_x) \!= &h(\delta) + (1 - \delta) \log 2\\ 
		&- H\!\left(d_s \!-\! \delta/2,\, d_x \!-\! d_s \!+\! \delta/2,\, 1 \!-\! d_x\right),
		\end{aligned}
		\end{equation}
		which is achieved by
		\begin{equation}
			P_{Z^{\star}Y^{\star}|X}(z,y|x) = P_{Z^{\star}Y^{\star}}(z, y)g(x)c(z,y,x),
		\end{equation}
		where
		\begin{equation}
		\label{D4_optimal_edge}
		P_{Z^{\star}Y^{\star}}(z, y) = \left\{
		\begin{array}{rcl}
		\frac{\delta + d_x -1}{\delta + 4 d_x - 2 d_s - 2} & & {z=y}\\
		\frac{d_x - d_s - \delta/2}{\delta + 4d_x - 2 d_s -2} & & {y=e}\\
		0 & & {\mbox{otherwise}}
		\end{array} \right.,
		\end{equation}
		\begin{equation}
		g(x) = \left\{
		\begin{array}{rcl}
		\frac{2(1 - d_x)}{1 - \delta} & & {x=0,1}\\
		\frac{1- d_x}{\delta} & & {x=e}
		\end{array} \right.,
		\end{equation}
		\begin{equation}
		c(z,y,x) = \left\{
		\begin{array}{rcl}
		1 & & {(z,y,x)\in \mathcal{G}_1}\\
		\frac{d_s - \delta/2}{1 - d_x} & & {(z,y,x)\in\mathcal{G}_2}\\
		\frac{\delta/2 + d_x - d_s}{1 - d_x} & & {(z,y,x)\in \mathcal{G}_3}\\
		0 & & {\mbox{otherwise}}
		\end{array} \right.,
		\end{equation}
		and sets $\mathcal{G}_1 = \{(0,0,0),(0,e,e),(1,1,1), (1,e,e)\}$, $\mathcal{G}_2 = \{(0,0,1),(0,e,1),(1,1,0), (1,e,0)\}$, $\mathcal{G}_3 = \{(0,0,e),(0,e,0),(1,1,e), (1,e,1)\}$.
		
		\item [(v)] if $(d_s, d_x) \in \mathcal{D}_5$, $R(d_s,d_x) = 0$,
	\end{itemize}
	where $x,y \in \{0,1,e\}$, $z \in \{0,1\}$, $h(\cdot)$ is the binary entropy function, $H(\cdot,\cdot,\cdot)$ is the discrete entropy function, and $0\log(0)$ is taken to be $0$.
	
\end{theorem}

\begin{proof}[Proof]
	We only give a sketch of the proof. 
	
	The rate-distortion function and the optimal test channel are obtained by solving problem \eqref{rate_distortion_function}. By the chain rule of mutual information, $I(X;Z,Y) \geq I(X;Y)$, with the equality hold if and only if $X-Y-Z$ form a Markov chain in this order. As a consequence, if there exist $P_{Y^{\star}|X}$ and $P_{Z^{\star}|Y^{\star}}$ such that $P_{Y^{\star}|X} = \arg\min_{P_{Y|X}}I(X;Y)\
	\mathrm{s.t.}\ \mathbb{E}\left[\textsf{d}_x(X,Y)\right] \leq d_x$, $X-Y^{\star}-Z^{\star}$ form a Markov chain, and $\mathbb{E}\left[\bar{\textsf{d}}_s(X,Z^{\star})\right] \leq d_s$, we have the optimal test channel $P_{Z^{\star}Y^{\star}|X} =P_{Z^{\star}|Y^{\star}} P_{Y^{\star}|X}$. This corresponds to the case when $(d_s, d_x) \in \mathcal{D}_1 \cup \mathcal{D}_2$. Similarly, $(d_s, d_x) \in \mathcal{D}_3$ implies the case when $X-Z^{\star}-Y^{\star}$ form a Markov chain. Given that $\mathbb{P}[Z^{\star}=0, Y^{\star}=0] = 1$ leads to the rate-distortion function when $(d_s, d_x) \in \mathcal{D}_5$, our focus lies solely on the case when $(d_s, d_x) \in \mathcal{D}_4$.
	
	The optimal test channel for $(d_s, d_x) \in \mathcal{D}_4$ is obtained by solving the Karush-Kuhn-Tucker (KKT) conditions of problem \eqref{rate_distortion_function}. Specifically, the Lagrangian associated with problem \eqref{rate_distortion_function} is written as
	\begin{align}
		&\mathcal{L}\left(\{P_{ZY|X}(z,y|x): (z,y,x) \in \mathcal{Q}\} , s_1, s_2, \lambda(x), \mu(x,z,y)\right) \nonumber \\ 
		=&\sum_{(z,y,x) \in \mathcal{Q}}P_X P_{ZY|X} \log \frac{P_{ZY|X}}{\sum_{x' \in \mathcal{X}}P_X P_{ZY|X}}\nonumber\\ 
		&+ s_1 \bigg(\sum_{(z,y,x) \in \mathcal{Q}}P_X P_{ZY|X}\bar{\textsf{d}}_s(x,z) - d_s \bigg) \nonumber \\
		&+ s_2 \bigg(\sum_{(z,y,x) \in \mathcal{Q}}P_X P_{ZY|X}\textsf{d}_x(x,y) - d_x \bigg)\nonumber\\ 
		&- \sum_{(z,y,x) \in \mathcal{Q}} \mu(x,z,y) P_{ZY|X}\nonumber\\ 
		&+ \sum_{x \in \mathcal{X}} \lambda(x)\bigg(\sum_{z \in \widehat{\mathcal{M}}, y \in \widehat{\mathcal{X}}}P_{ZY|X} - 1\bigg),
	\end{align}
	where $\mathcal{Q} \triangleq \{(z,y,x): P_{ZY|X} > 0 \}$ and the Lagrange multipliers $s_1 \geq 0$, $s_2 \geq 0$, $\mu(x,z,y) \geq 0$. Clearly, varying selections of $\mathcal{Q}$ lead to different Lagrangians and, consequently, different KKT conditions. The optimal test channel for $(d_s, d_x) \in \mathcal{D}_4$ is obtained by solving the KKT conditions correspond to $\mathcal{Q} = \{(z,y,x): (z,y) \neq (0, 1),  (z,y) \neq (1, 0)\}$.

\end{proof}

\begin{figure}[htbp]
	\centering
	\includegraphics[width=1\linewidth]{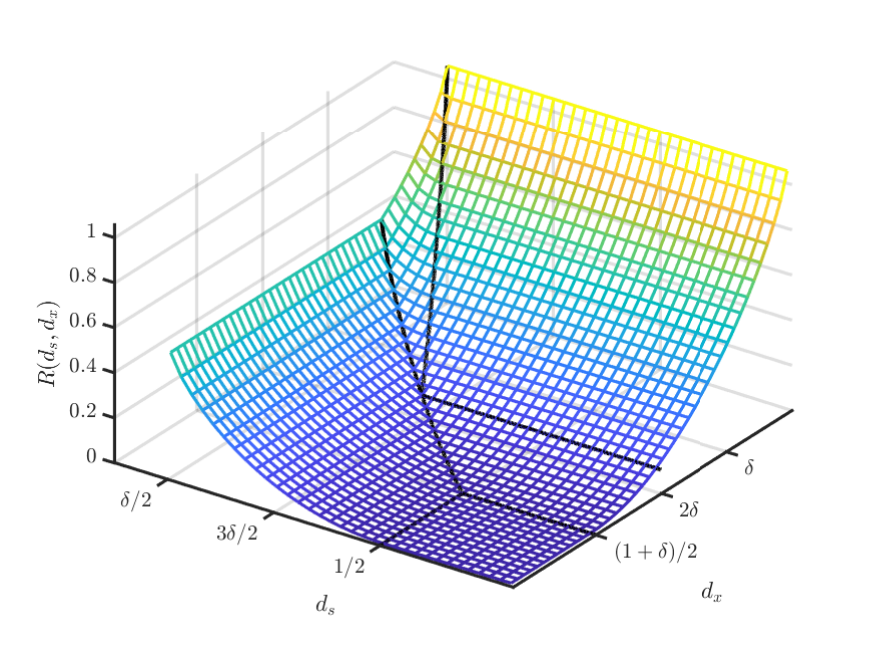}
	\caption{Rate-distortion function of the EFCF case with $\delta = 0.2$.}
	\label{R_D}
\end{figure}

The rate-distortion function with $\delta=0.2$ is plotted in Fig. \ref{R_D}. For $(d_s, d_x) \in \mathcal{D}_4$, some useful quantities are computed as follows:
\begin{equation}
\lambda_s^{\star} = \log\left(\frac{d_x - d_s + \delta / 2}{d_s - \delta/2}\right),
\end{equation}
\begin{equation}
\lambda_x^{\star} = \log\left(\frac{1 - d_x}{d_x - d_s + \delta / 2}\right),
\end{equation}

\begin{align}
	\jmath_{X}(0,d_s,d_x) =&\jmath_{X}(1,d_s,d_x) \nonumber \\ 
	=& -\lambda_s^{\star}d_s - \lambda_x^{\star} d_x - \log\! \left(\!\frac{1 \!-\! \delta}{2 (1 \!-\! d_x)}\!\right)
\end{align}
\begin{equation}
	\jmath_{X}(e,d_s,d_x) =-\lambda_s^{\star}d_s - \lambda_x^{\star} d_x - \log\!\left(\!\sqrt{\frac{d_s \!-\! \delta/2}{d_x \!-\! d_s \!+\! \delta/2}}\!\cdot\! \frac{\delta}{1 \!-\! d_x}\!\right)
\end{equation}
%\begin{equation}
%\jmath_{X}(x,d_s,d_x) = \left\{
%\begin{array}{rcl}
%-\lambda_s^{\star}d_s - \lambda_x^{\star} d_x - \log\left(\frac{1 - \delta}{2 (1 - d_x)}\right) & & {x=0,1}\\
%-\lambda_s^{\star}d_s - \lambda_x^{\star} d_x - \log\left(\sqrt{\frac{d_s - \delta/2}{d_x - d_s + \delta/2}}\cdot \frac{\delta}{1 - d_x}\right) & & {x=e}
%\end{array} \right.,
%\end{equation}
\begin{equation}
	\mathcal{V}(d_s,d_x) = \delta  (1 - \delta) \log^2 \left(\sqrt{\frac{d_s \!-\! \delta/2}{d_x \!-\! d_s \!+\! \delta/2}} \cdot \frac{2\delta}{1\!-\!\delta} \right),
\end{equation}
\begin{equation}
	\tilde{\mathcal{V}}(d_s,d_x) = \mathcal{V}(d_s,d_x) + \frac{\delta}{4} \left(\lambda_s^{\star}\right)^2.
\end{equation}

\subsection{Nonasymptotic Converse Bound}

\begin{theorem} 
	\label{case_study_converse}
	(Converse, EFCF): In erased fair coin flips, any $(k, M, d_s, d_x, \epsilon)$ code must satisfy
	\begin{equation}
	\label{case_study_converse_bound}
	\begin{aligned}
		\epsilon \geq& \sup_{\gamma \geq 0} \bigg\{\mathbb{P}\bigg[\sum_{i=1}^k \jmath_{X}(X_i,d_s,d_x) + \lambda_s^{\star} \big(\textsf{d}_s(S_i, 0)\\ 
		&- \bar{\textsf{d}}_s(X_i, 0)\big) \geq \gamma + \log M\bigg] - \exp(-\gamma)\bigg\},
	\end{aligned}
	\end{equation}
	where 
%	\begin{equation}
%	\textsf{d}_s(s, z)=\left\{
%	\begin{array}{rcl}
%	0 & & {s=z}\\
%	1 & & {s\neq z}
%	\end{array} \right.,
%	\
%	\bar{\textsf{d}}_s(x, z)=\left\{
%	\begin{array}{rcl}
%	0 & & {x=z}\\
%	1 & & {x \neq z,\ x \neq e}\\
%	1/2 & & {x=e}
%	\end{array} \right.,
%	\end{equation}
	\begin{equation}
	\label{EFCF_distortion_measures}
	\textsf{d}_s(s, z)\!=\!\textrm{1}\{s \!\neq\! z\},
	\
	\bar{\textsf{d}}_s(x, z)\!=\!\left\{\!\!
	\begin{array}{rcl}
	\textrm{1}\{x \!\neq\! z\} & & {x \neq e}\\
	1/2 & & {x=e}
	\end{array} \right.\!\!,
	\end{equation}
	and $\{(S_i, X_i)\}_{i=1}^k$ are independently drawn from the EFCF source.
\end{theorem}
\begin{proof}[Proof]
We weaken the bound in \eqref{converse_bound2} by choosing
\begin{align}
P_{\bar{X}^k|\bar{Z}^k=z^k,\bar{Y}^k=y^k}(x^k)=&\prod_{i=1}^k P_{X|Z^{\star}=z_i, Y^{\star}=y_i} (x_i), \\
\lambda_s=& k \lambda_s^{\star},\\
\lambda_x=& k \lambda_x^{\star}.
\end{align}
By Corollary \ref{converse_bound_corollary}, any $(k, M, d_s, d_x, \epsilon)$ code must satisfy
\begin{align} 
\epsilon \geq& \sup_{\gamma \geq 0} \bigg\{\mathbb{E}\bigg[\min_{z^k \in \hat{\mathcal{S}}^k, y^k \in \hat{\mathcal{A}}^k}
\mathbb{P}\bigg[\sum_{i=1}^k\imath_{X;Z^\star Y^\star}(X_i;z_i,y_i)\nonumber\\
&+ \lambda_s^{\star}(\textsf{d}_s(S_i,z_i) - d_s) +  \lambda_x^{\star} (\textsf{d}_x(X_i,y_i) - d_x)\nonumber\\
&\geq \log M + \gamma| X^k \bigg]\bigg] - \exp(-\gamma)\bigg\}\\
=& \sup_{\gamma \geq 0} \bigg\{\mathbb{E}\bigg[\min_{z^k \in \hat{\mathcal{S}}^k}
\mathbb{P}\bigg[\sum_{i=1}^k\jmath_{X}(X_i,d_s,d_x) + \lambda_s^\star \big( \textsf{d}_s(S_i,z_i)\nonumber \\ 
&- \bar{\textsf{d}}_s(X_i,z_i) \big)  
\geq \log M + \gamma| X^k \bigg]\bigg] - \exp(-\gamma)\bigg\}, \label{EFCF_converse_temp1}
\end{align}
where \eqref{EFCF_converse_temp1} is by \eqref{relation_of_tilted_information}. By \eqref{EFCF_distortion_measures}, the probability in \eqref{EFCF_converse_temp1} is the same for all $z^k \in \hat{\mathcal{S}}^k$, leading to \eqref{case_study_converse_bound}.
\end{proof}

\subsection{Nonasymptotic Achievability Bounds}
%Due to space limitations, we provide the achievability bound solely for the most critical case, $(d_s, d_x) \in \mathcal{D}_4$, where both inequality constraints are satisfied with equalities at the optimal point. The achievability bounds for other cases can be derived in a similar manner.

We first give the nonasymptotic achievability bound for $(d_s, d_x) \in \mathcal{D}_1 \cup \mathcal{D}_4$.
\begin{theorem} 
	\label{case_study_achievability}
	(Achievability, EFCF with $(d_s, d_x) \in \mathcal{D}_1 \cup \mathcal{D}_4$): In erased fair coin flips, for $(d_s, d_x) \in \mathcal{D}_1 \cup \mathcal{D}_4$, there exists an $(k, M, d_s, d_x, \epsilon)$ code such that 
	\begin{align}
	\label{case_study_achievability_bound}
		\epsilon \leq &\sum_{t=0}^k \mathsf{binopmf}(t; k, \delta) \cdot \sum_{i=0}^t \mathsf{binopmf}(i; t, 1/2)\nonumber\\ 
		&\cdot \bigg(1 - \sum_{j=0}^t \mathsf{binopmf}(j; t, \mathbb{P}[Y^{\star} \neq e]) \nonumber \\ 
		&\cdot \sum_{r=0}^{k-t} \mathsf{binopmf}(r; k - t, \mathbb{P}[Y^{\star} = e])\nonumber\\ 
		&\cdot \sum_{v=0}^{\lfloor k d_x \rfloor - j} \mathsf{binopmf}(v - r; k - t - r, 1/2) \nonumber \\ 
		&\cdot \mathsf{binocdf}(\lfloor k d_s \rfloor - i -(v-r); r, 1/2) \bigg)^M,
	\end{align} 
%	\begin{equation}
%		\mathbb{P}[\tilde{Y} = e] = \frac{2 d_x - 2 d_s - \delta}{\delta + 4d_x - 2 d_s -2},
%	\end{equation}
%	\begin{equation}
%	\mathbb{P}[\tilde{Y} \neq e] = \frac{2 \delta + 2 d_x -2}{\delta + 4d_x - 2 d_s -2},
%	\end{equation}
	where $\mathsf{binopmf}(\cdot; n, p)$ and $\mathsf{binocdf}(\cdot; n, p)$ denote the probability mass function (PMF) and the cumulative distribution function (CDF) of the binomial distribution, respectively, with $n$ degrees of freedom and success probability $p$.
	
\end{theorem}

\begin{figure}[htbp]
	\centering
	\includegraphics[width=1\linewidth]{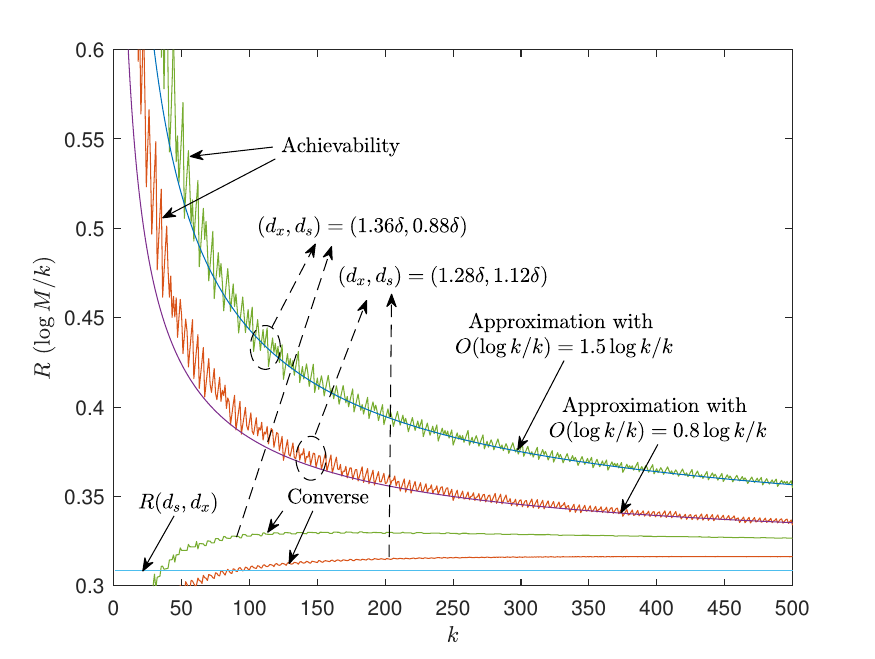}
	\caption{Rate-blocklength trade-off in the erased fair coin flips case with $\delta = 0.2$ and $\epsilon = 0.1$. Note that $(d_x, d_s) = (1.36\delta, 0.88\delta)$ and $(d_x, d_s) = (1.28\delta, 1.12\delta)$ share the same value of $R(d_s,d_x)$.}
	\label{convergence_result}
\end{figure}

\begin{proof}[Proof]

We adopt the achievability scheme in \cite[Theorem 1]{Yang2024Joint}. Specifically, provided $M$ codewords $\{(\tilde{z}_m^k, \tilde{y}_m^k)\}_{m=1}^M$ and having observed $x^k$, the encoder $\textsf{f}$ chooses 
\begin{equation}
\begin{aligned}
m^\star \in \arg\min_{m} \mathbb{P}[&\textsf{d}_s(S^k, \tilde{z}_m^k) > d_s\\ 
&\cup \textsf{d}_x(X^k, \tilde{y}_m^k) > d_x | X^k=x^k],
\end{aligned}
\end{equation}
i.e., $\textsf{f}(x^k) = m^\star$. Then, the decoder outputs $\textsf{c}(\textsf{f}(x^k)) = (\tilde{z}_m^k, \tilde{y}_m^k)$. We separate the decoder into $\textsf{c}_s(\textsf{f}(x^k)) = \tilde{z}_m^k$ and $\textsf{c}_x(\textsf{f}(x^k)) = \tilde{y}_m^k$ for convenience.

Consider the ensemble of codes with codewords
drawn i.i.d. from the optimal edge distribution $P_{Z^\star Y^\star}$ for $(d_s, d_x) \in \mathcal{D}_1 \cup \mathcal{D}_4$. We discuss cases $(d_s, d_x) \in \mathcal{D}_1$ and $(d_s, d_x) \in \mathcal{D}_4$ together because both cases have $P_{Z^\star Y^\star}$ with the same structure. Specifically, given $(d_s, d_x) \in \mathcal{D}_1 \cup \mathcal{D}_4$, let $P_{\tilde{Z}\tilde{Y}}$ be the corresponding optimal edge distribution $P_{Z^\star Y^\star}$ in Theorem \ref{theorem_Rate_Distortion_func}, and define $M$ mutually independent random codewords $\{(\tilde{Z}_m^k, \tilde{Y}_m^k) \sim P_{\tilde{Z}\tilde{Y}} \times \dots \times P_{\tilde{Z}\tilde{Y}}\}_{m=1}^M$. Let $(\tilde{Z}^k, \tilde{Y}^k) \sim P_{\tilde{Z}\tilde{Y}} \times \dots \times P_{\tilde{Z}\tilde{Y}}$. 
%We adopt the achievability scheme used to prove \cite[Theorem 1]{Yang2024Joint}.
Then, the ensemble error probability can be computed by \eqref{ensemble_error_prob_01}-\eqref{achi_1},
\begin{figure*}[!t]
	\begin{align}
	&\mathbb{P}\left[\textsf{d}_s(S^k, \textsf{C}_s(\textsf{f}(X^k))) > d_s \cup \textsf{d}_x(X^k,\textsf{C}_x(\textsf{f}(X^k))) > d_x\right] \label{ensemble_error_prob_01} \\
	=&\sum_{t=0}^k \mathbb{P}\left[t\ \text{erasures}\ \text{in}\ X^k \right] \cdot \mathbb{P}\big[ \textsf{d}_s(S^k, \textsf{C}_s(\textsf{f}(X^k))) > d_s\cup \textsf{d}_x(X^k,\textsf{C}_x(\textsf{f}(X^k))) > d_x | t\ \text{erasures}\ \text{in}\ X^k \big]\\
	=&\sum_{t=0}^k \mathbb{P}\left[t\ \text{erasures}\ \text{in}\ X^k \right]\cdot \mathbb{P}\big[ k\textsf{d}_s(S^k, \textsf{C}_s(\textsf{f}(X^k))) \!>\! \lfloor k d_s \rfloor\cup k \textsf{d}_x(X^k,\textsf{C}_x(\textsf{f}(X^k))) \!>\! \lfloor k d_x \rfloor | t\ \text{erasures}\ \text{in}\ X^k \big]\\
	=&\sum_{t=0}^k \mathbb{P}\left[t\ \text{erasures}\ \text{in}\ X^k \right] \cdot \sum_{i=0}^t \bigg \{ \mathbb{P}\left[t \textsf{d}_s(S^t,\textsf{C}_s(\textsf{f}(X^t))) = i | X^t = e,\dots, e \right] \nonumber \\
	&\cdot \mathbb{P}\big[ (k-t)\textsf{d}_s(S^{k-t}, \textsf{C}_s(\textsf{f}(X^{k-t}))) \!>\! \lfloor k d_s \rfloor - i \cup k \textsf{d}_x(X^{k},\textsf{C}_x(\textsf{f}(X^{k}))) \!>\! \lfloor k d_x \rfloor | X^{k-t} \!=\! S^{k-t},t\ \text{erasures}\ \text{in}\ X^k  \big] \bigg \} \label{achi_1}
	\end{align}
	\hrulefill
\end{figure*}
%\begin{align}
%	&\mathbb{P}\left[\textsf{d}_s(S^k, \textsf{C}_s(\textsf{f}(X^k))) > d_s \cup \textsf{d}_x(X^k,\textsf{C}_x(\textsf{f}(X^k))) > d_x\right]\\
%	=&\sum_{t=0}^k \mathbb{P}\left[t\ \text{erasures}\ \text{in}\ X^k \right] \cdot \mathbb{P}\big[ \textsf{d}_s(S^k, \textsf{C}_s(\textsf{f}(X^k))) > d_s \nonumber \\ 
%	&\cup \textsf{d}_x(X^k,\textsf{C}_x(\textsf{f}(X^k))) > d_x | t\ \text{erasures}\ \text{in}\ X^k \big]\\
%	=&\sum_{t=0}^k \mathbb{P}\left[t\ \text{erasures}\ \text{in}\ X^k \right]\cdot \mathbb{P}\big[ k\textsf{d}_s(S^k, \textsf{C}_s(\textsf{f}(X^k))) \!>\! \lfloor k d_s \rfloor\nonumber \\  
%	&\cup k \textsf{d}_x(X^k,\textsf{C}_x(\textsf{f}(X^k))) \!>\! \lfloor k d_x \rfloor | t\ \text{erasures}\ \text{in}\ X^k \big]\\
%	=&\sum_{t=0}^k \mathbb{P}\left[t\ \text{erasures}\ \text{in}\ X^k \right] \nonumber \\ 
%	&\cdot \sum_{i=0}^t \bigg \{ \mathbb{P}\left[t \textsf{d}_s(S^t,\textsf{C}_s(\textsf{f}(X^t))) = i | X^t = e,\dots, e \right] \nonumber \\
%	&\cdot \mathbb{P}\big[ (k-t)\textsf{d}_s(S^{k-t}, \textsf{C}_s(\textsf{f}(X^{k-t}))) > \lfloor k d_s \rfloor - i \nonumber \\ 
%	&\cup k \textsf{d}_x(X^{k},\textsf{C}_x(\textsf{f}(X^{k}))) > \lfloor k d_x \rfloor | X^{k-t} = S^{k-t}\nonumber \\
%	&,t\ \text{erasures}\ \text{in}\ X^k  \big] \bigg \}, \label{achi_1}
%\end{align}
where \eqref{achi_1} holds since the distribution of $k\textsf{d}_s(S^t, \textsf{C}_s(\textsf{f}(X^t)))$, given $ X^t = e,\dots, e$, does not dependent on the codebook. 

Next, we compute the probability terms in \eqref{achi_1}. Clearly, since $\mathbb{P}\left[X=e \right] = \delta$,
\begin{align}
	\mathbb{P}\left[t\ \text{erasures}\ \text{in}\ X^k \right] = C_{k}^t \cdot \delta^t \cdot (1 - \delta)^{k-t};
\end{align}
since $P_{S|X}(0|e) = P_{S|X}(1|e) = 1/2$,
\begin{align}
	\mathbb{P}\left[ t\textsf{d}_s(S^t, \textsf{C}_s(\textsf{f}(X^t))) = i| X^t = e,\dots, e \right] = C_{t}^{i} \cdot \frac{1}{2^t}.
\end{align}
For the computation of the last probability term in \eqref{achi_1}, we have \eqref{last_term_01}-\eqref{last_term_04},
\begin{figure*}[!t]
	\begin{align}
	&\mathbb{P}\left[ (k-t)\textsf{d}_s(S^{k-t}, \textsf{C}_s(\textsf{f}(X^{k-t}))) > \lfloor k d_s \rfloor - i \cup k \textsf{d}_x(X^{k},\textsf{C}_x(\textsf{f}(X^{k}))) > \lfloor k d_x \rfloor | X^{k-t} = S^{k-t},\ t\ \text{erasures}\ \text{in}\ X^k  \right] \label{last_term_01} \\
	=&\prod_{m=1}^M \mathbb{P}\left[ (k-t)\textsf{d}_s(S^{k-t}, \tilde{Z}^{k-t}_m) > \lfloor k d_s \rfloor - i \cup k \textsf{d}_x(X^k, \tilde{Y}^k_m) > \lfloor k d_x \rfloor | X^{k-t} = S^{k-t},\ t\ \text{erasures}\ \text{in}\ X^k  \right]\\
	=&\left( \mathbb{P}\left[ (k-t)\textsf{d}_s(S^{k-t}, \tilde{Z}^{k-t}) > \lfloor k d_s \rfloor - i \cup k \textsf{d}_x(X^{k}, \tilde{Y}^{k}) > \lfloor k d_x \rfloor | X^{k-t} = S^{k-t},\ t\ \text{erasures}\ \text{in}\ X^k  \right] \right)^M \\
	=&\left(1 - \mathbb{P}\left[ (k-t)\textsf{d}_s(S^{k-t}, \tilde{Z}^{k-t}) \leq \lfloor k d_s \rfloor - i \cap k \textsf{d}_x(X^k, \tilde{Y}^k) \leq \lfloor k d_x \rfloor | X^{k-t} = S^{k-t},\ t\ \text{erasures}\ \text{in}\ X^k \right] \right)^M \label{last_term_04}
	\end{align}
	\hrulefill
\end{figure*}
%\begin{align}
%	&\mathbb{P}\left[ (k-t)\textsf{d}_s(S^{k-t}, \textsf{C}_s(\textsf{f}(X^{k-t}))) > \lfloor k d_s \rfloor - i\right. \nonumber \\ 
%	&\left. \cup k \textsf{d}_x(X^{k},\textsf{C}_x(\textsf{f}(X^{k}))) > \lfloor k d_x \rfloor | X^{k-t} = S^{k-t},\ t\ \text{erasures}\ \text{in}\ X^k  \right]\\
%	=&\prod_{m=1}^M \mathbb{P}\left[ (k-t)\textsf{d}_s(S^{k-t}, \tilde{Z}^{k-t}_m) > \lfloor k d_s \rfloor - i \cup k \textsf{d}_x(X^k, \tilde{Y}^k_m) > \lfloor k d_x \rfloor | X^{k-t} = S^{k-t},\ t\ \text{erasures}\ \text{in}\ X^k  \right]\\
%	=&\left( \mathbb{P}\left[ (k-t)\textsf{d}_s(S^{k-t}, \tilde{Z}^{k-t}) > \lfloor k d_s \rfloor - i \cup k \textsf{d}_x(X^{k}, \tilde{Y}^{k}) > \lfloor k d_x \rfloor | X^{k-t} = S^{k-t},\ t\ \text{erasures}\ \text{in}\ X^k  \right] \right)^M \\
%	=&\left(1 - \mathbb{P}\left[ (k-t)\textsf{d}_s(S^{k-t}, \tilde{Z}^{k-t}) \leq \lfloor k d_s \rfloor - i \cap k \textsf{d}_x(X^k, \tilde{Y}^k) \leq \lfloor k d_x \rfloor | X^{k-t} = S^{k-t},\ t\ \text{erasures}\ \text{in}\ X^k \right] \right)^M,
%\end{align}
where the probability term in \eqref{last_term_04} can be computed by \eqref{prob_term_01}-\eqref{prob_term_04},
\begin{figure*}[!t]
	\begin{align}
	&\mathbb{P}\left[ (k-t)\textsf{d}_s(S^{k-t}, \tilde{Z}^{k-t}) \leq \lfloor k d_s \rfloor - i \cap k \textsf{d}_x(X^k, \tilde{Y}^k) \leq \lfloor k d_x \rfloor | X^{k-t} = S^{k-t},\ t\ \text{erasures}\ \text{in}\ X^k \right] \label{prob_term_01} \\
	=&\sum_{j=0}^t \mathbb{P}\left[ t \textsf{d}_x(X^t, \tilde{Y}^t) = j | X^t = e,\dots, e \right] \nonumber \\ 
	&\cdot \mathbb{P}\left[ (k-t)\textsf{d}_s(S^{k-t}, \tilde{Z}^{k-t}) \leq \lfloor k d_s \rfloor - i \cap (k - t) \textsf{d}_x(X^{k - t}, \tilde{Y}^{k-t}) \leq \lfloor k d_x \rfloor - j | X^{k-t} = S^{k-t}\right]\\
	=&\sum_{j=0}^t \bigg \{ \mathbb{P}\left[ t \textsf{d}_x(X^t, \tilde{Y}^t) = j | X^t = e,\dots, e \right] \cdot \sum_{r = 0}^{k-t} \mathbb{P}\left[ \text{$r$ erasures in $\tilde{Y}^{k-t}$} \right] \nonumber \\
	&\cdot \mathbb{P}\left[ (k-t)\textsf{d}_s(S^{k-t}, \tilde{Z}^{k-t}) \!\leq\! \lfloor k d_s \rfloor - i \cap (k - t) \textsf{d}_x(X^{k - t}, \tilde{Y}^{k-t}) \!\leq\! \lfloor k d_x \rfloor - j | X^{k-t} \!=\! S^{k-t}, \text{$r$ erasures in $\tilde{Y}^{k-t}$} \right] \bigg \}\\
	=&\sum_{j=0}^t \bigg \{ \mathbb{P}\left[ t \textsf{d}_x(X^t, \tilde{Y}^t) = j | X^t = e,\dots, e \right] \cdot \sum_{r = 0}^{k-t} \mathbb{P}\left[ \text{$r$ erasures in $\tilde{Y}^{k-t}$} \right] \nonumber \\
	&\cdot \sum_{v=0}^{\lfloor k d_x \rfloor - j} \mathbb{P}\left[ (k\!-\!t)\textsf{d}_s(S^{k-t}, \tilde{Z}^{k-t}) \!\leq\! \lfloor k d_s \rfloor - i \cap (k - t) \textsf{d}_x(X^{k - t}, \tilde{Y}^{k-t}) =v | X^{k-t} \!=\! S^{k-t}, \text{$r$ erasures in $\tilde{Y}^{k-t}$} \right] \bigg \} \label{prob_term_04}
	\end{align}
	\hrulefill
\end{figure*}
%\begin{align}
%	&\mathbb{P}\left[ (k-t)\textsf{d}_s(S^{k-t}, \tilde{Z}^{k-t}) \leq \lfloor k d_s \rfloor - i \cap k \textsf{d}_x(X^k, \tilde{Y}^k) \leq \lfloor k d_x \rfloor | X^{k-t} = S^{k-t},\ t\ \text{erasures}\ \text{in}\ X^k \right]\\
%	=&\sum_{j=0}^t \mathbb{P}\left[ t \textsf{d}_x(X^t, \tilde{Y}^t) = j | X^t = e,\dots, e \right] \nonumber \\ 
%	&\cdot \mathbb{P}\left[ (k-t)\textsf{d}_s(S^{k-t}, \tilde{Z}^{k-t}) \leq \lfloor k d_s \rfloor - i \cap (k - t) \textsf{d}_x(X^{k - t}, \tilde{Y}^{k-t}) \leq \lfloor k d_x \rfloor - j | X^{k-t} = S^{k-t}\right]\\
%	=&\sum_{j=0}^t \bigg \{ \mathbb{P}\left[ t \textsf{d}_x(X^t, \tilde{Y}^t) = j | X^t = e,\dots, e \right] \cdot \sum_{r = 0}^{k-t} \mathbb{P}\left[ \text{$r$ erasures in $\tilde{Y}^{k-t}$} \right] \nonumber \\
%	&\cdot \mathbb{P}\left[ (k-t)\textsf{d}_s(S^{k-t}, \tilde{Z}^{k-t}) \leq \lfloor k d_s \rfloor - i \right. \nonumber \\
%	&\left. \cap (k - t) \textsf{d}_x(X^{k - t}, \tilde{Y}^{k-t}) \leq \lfloor k d_x \rfloor - j | X^{k-t} = S^{k-t}, \text{$r$ erasures in $\tilde{Y}^{k-t}$} \right] \bigg \}\\
%	=&\sum_{j=0}^t \bigg \{ \mathbb{P}\left[ t \textsf{d}_x(X^t, \tilde{Y}^t) = j | X^t = e,\dots, e \right] \cdot \sum_{r = 0}^{k-t} \mathbb{P}\left[ \text{$r$ erasures in $\tilde{Y}^{k-t}$} \right] \nonumber \\
%	&\cdot \sum_{v=0}^{\lfloor k d_x \rfloor - j} \mathbb{P}\left[ (k-t)\textsf{d}_s(S^{k-t}, \tilde{Z}^{k-t}) \leq \lfloor k d_s \rfloor - i \right. \nonumber \\
%	&\left. \cap (k - t) \textsf{d}_x(X^{k - t}, \tilde{Y}^{k-t}) =v | X^{k-t} = S^{k-t}, \text{$r$ erasures in $\tilde{Y}^{k-t}$} \right] \bigg \},
%\end{align}
with
\begin{equation}
	\begin{aligned}
	&\mathbb{P}\left[ t \textsf{d}_x(X^t, \tilde{Y}^t) = j | X^t = e,\dots, e \right]\\ 
	&= C_t^j \cdot (\mathbb{P}[\tilde{Y} \neq e])^j \cdot (\mathbb{P}[\tilde{Y} = e])^{t-j},
	\end{aligned}
\end{equation}
\begin{equation}
	\begin{aligned}
	&\mathbb{P}\left[ \text{$r$ erasures in $\tilde{Y}^{k-t}$} \right]\\ 
	&= C_{k-t}^r \cdot (\mathbb{P}[\tilde{Y} = e])^r \cdot (\mathbb{P}[\tilde{Y} \neq e])^{k-t-r},
	\end{aligned}
\end{equation}
and the last probability term in \eqref{prob_term_04} being computed by \eqref{Last_Prob_01}-\eqref{Last_Prob_03}.
\begin{figure*}[!t]
	\begin{align}
	&\mathbb{P}\left[ (k-t)\textsf{d}_s(S^{k-t}, \tilde{Z}^{k-t}) \leq \lfloor k d_s \rfloor - i  \cap (k - t) \textsf{d}_x(X^{k - t}, \tilde{Y}^{k-t}) = v | X^{k-t} = S^{k-t}, \text{$r$ erasures in $\tilde{Y}^{k-t}$} \right] \label{Last_Prob_01} \\
	=& \mathbb{P}\left[ \text{The $(k-t-r)$ non-erasure elements in $\tilde{Y}^{k-t}$ intorduce $(v-r)$ bit errors}  \right] \nonumber \\
	&\cdot \mathbb{P}\Big[ \text{The $r$ elements in $\tilde{Z}^{k-t}$ corresponding to the $r$ erasures in $\tilde{Y}^{k-t}$ intorduce}\nonumber\\ 
	&\text{ less than or equal to $(\lfloor k d_s \rfloor \!-\! i\!-\!(v\!-\!r))$ bit errors}  \Big] \label{language_description} \\
	=&\left\{
	\begin{array}{rcl}
	C_{k-t-r}^{v-r} \cdot \frac{1}{2^{k-t-r}} \cdot \sum_{u=0}^{\lfloor k d_s \rfloor - i -(v-r)} C_{r}^{u} \cdot \frac{1}{2^{r}} & & {
		\begin{array}{l}
		v-r \in \{0, \dots, k-t-r\}\\
		\text{and}\ \lfloor k d_s \rfloor - i - (v-r) \in \{0, \dots, r\}
		\end{array}
	}  \\
	C_{k-t-r}^{v-r} \cdot \frac{1}{2^{k-t-r}} & & {
		\begin{array}{l}
		v-r \in \{0, \dots, k-t-r\}\\
		\text{and}\ \lfloor k d_s \rfloor - i - (v-r) > r
		\end{array}
	}  \\
	0 & & {\text{otherwise}}
	\end{array} \right. \label{Last_Prob_03}
	\end{align}
	\hrulefill
\end{figure*}
%\begin{align}
%	&\mathbb{P}\left[ (k-t)\textsf{d}_s(S^{k-t}, \tilde{Z}^{k-t}) \leq \lfloor k d_s \rfloor - i  \cap (k - t) \textsf{d}_x(X^{k - t}, \tilde{Y}^{k-t}) = v | X^{k-t} = S^{k-t}, \text{$r$ erasures in $\tilde{Y}^{k-t}$} \right]\\
%	=& \mathbb{P}\left[ \text{The $(k-t-r)$ non-erasure elements in $\tilde{Y}^{k-t}$ intorduce $(v-r)$ bit errors}  \right] \nonumber \\
%	&\cdot \mathbb{P}\Big[ \text{The $r$ elements in $\tilde{Z}^{k-t}$ corresponding to the $r$ erasures in $\tilde{Y}^{k-t}$} \nonumber  \\ 
%	&\text{intorduce less than or equal to $(\lfloor k d_s \rfloor - i-(v-r))$ bit errors}  \Big] \label{language_description} \\
%	=&\left\{
%	\begin{array}{rcl}
%	C_{k-t-r}^{v-r} \cdot \frac{1}{2^{k-t-r}} \cdot \sum_{u=0}^{\lfloor k d_s \rfloor - i -(v-r)} C_{r}^{u} \cdot \frac{1}{2^{r}} & & {
%		\begin{array}{l}
%		v-r \in \{0, \dots, k-t-r\}\\
%		\text{and}\ \lfloor k d_s \rfloor - i - (v-r) \in \{0, \dots, r\}
%		\end{array}
%	}  \\
%	C_{k-t-r}^{v-r} \cdot \frac{1}{2^{k-t-r}} & & {
%		\begin{array}{l}
%		v-r \in \{0, \dots, k-t-r\}\\
%		\text{and}\ \lfloor k d_s \rfloor - i - (v-r) > r
%		\end{array}
%	}  \\
%	0 & & {\text{otherwise}}
%	\end{array} \right..
%\end{align}
Note that \eqref{language_description} can be obtained by realizing that $\mathbb{P}[\textsf{d}_x(X, \tilde{Y}) = 1|X=S, \tilde{Y}=e]=1$ and $\mathbb{P}[\textsf{d}_s(S, \tilde{Z}) = \textsf{d}_x(X, \tilde{Y})|X=S, \tilde{Y} \neq e]=1$. 

Since there exists at least one code with an excess probability no greater than the ensemble average, the code satisfying \eqref{case_study_achievability_bound} must exist.

\end{proof}

The nonasymptotic achievability bound for $(d_s, d_x) \in \mathcal{D}_2 \cup \mathcal{D}_3$ is given by the following theorem.
\begin{theorem} 
	\label{case_study_achievability2}
	(Achievability, EFCF with $(d_s, d_x) \in \mathcal{D}_2 \cup \mathcal{D}_3$): In erased fair coin flips, for $(d_s, d_x) \in \mathcal{D}_2 \cup \mathcal{D}_3$, there exists an $(k, M, d_s, d_x, \epsilon)$ code such that 
	\begin{align}
	\label{case_study_achievability2_bound}
	\epsilon \! \leq \! &\sum_{t=0}^k \mathsf{binopmf}(t; k, \delta) \cdot \sum_{i=0}^t \mathsf{binopmf}(i; t, 1/2)\nonumber\\ 
	&\cdot \bigg(\!1 - \mathsf{binocdf}(\min\{\lfloor k d_s \rfloor \!-\! i, \lfloor k d_x \rfloor \!-\! t\}; k\!-\!t, 1/2)\!\bigg)^M.
	\end{align} 
	
\end{theorem}
\begin{proof}[Proof]
	This proof closely resembles that of Theorem \ref{case_study_achievability} and is omitted here.
\end{proof}

\subsection{Numerical Results}
The second-order approximation in Theorem \ref{theorem_dispersion}, the nonasymptotic converse bound with $\gamma=\log k / 2$ in Theorem \ref{case_study_converse}, the nonasymptotic achievability bound in Theorem \ref{case_study_achievability}, and the asymptotically achievable rate \eqref{R_D_for_D4} are plotted in Fig. \ref{convergence_result} for both maximum admissible distortion settings $(d_x,d_s) = (1.36\delta, 0.88\delta)$ and $(d_x,d_s) = (1.28\delta, 1.12\delta)$ with $\delta = 0.2$ and $\epsilon = 0.1$. Note that the two pairs of maximum admissible distortions fall in $\mathcal{D}_4$ and share the same asymptotically achievable rate $R(d_s,d_x)$. From Fig. \ref{convergence_result}, we observe that the second-order approximations well approximate their corresponding finite-blocklength rates, requiring only slightly adapting the remainder term. 
%(the dispersion of EFCF is not large enough to cover the impact of the remainder term). 
Furthermore, despite both the $(d_s,d_x)$ settings share the same asymptotically achievable rate $R(d_s,d_x)$, we observe a noticeable gap between their finite-blocklength rates. By Theorem \ref{theorem_dispersion}, we know this gap mainly arises from their distinct dispersions, technically, distinct values of the rate-dispersion function $\tilde{\mathcal{V}}(d_s,d_x)$.

%associated with the two $(d_x,d_s)$ settings.
%
%with the dispersion being characterized by the noisy rate-distortion function \eqref{noisy_rate_dispersion_func}.

\section{Conclusion}
\label{SEC_6}

In this paper, we derived nonasymptotic and second-order converse bounds for the joint data and semantics lossy compression problem. The obtained second-order converse bound, coinciding with the second-order achievability bound in \cite{Yang2024Joint}, implies that the (excess distortion) dispersion of joint data and semantics lossy compression is $\tilde{\mathcal{V}}(d_s,d_x)$ for $(d_s, d_x) \in \textsf{int}(\mathcal{D}_{sx})$. Additionally, since the second-order bounds are derived through asymptotic analysis of the nonasymptotic bounds, their tightness implies that the nonasymptotic bounds proposed in this paper and \cite{Yang2024Joint} are asymptotically tight up to the second order. Finally, we obtained nonasymptotic achievability and converse bounds for the case of erased fair coin flips. Based on these results, we numerically illustrated the accuracy of the second-order approximation.

%asymptotic tightness up to the second order for the nonasymptotic bounds proposed in this paper and \cite{Yang2024Joint}.

\appendices

\section{Auxiliary Results}
\label{Auxiliary_Results}

%Following \cite{Kostina2016Nonasymptotic}, 
We consider $R_X(d_s,d_x) = R_{P_X}(d_s,d_x)$ as a function of the $|\mathcal{X}|$-dimensional probability vector $P_X$, and define
\begin{equation}
R_{Q_X}(d_s,d_x) \triangleq R_{P_{\bar{X}}}(d_s,d_x),
\end{equation}
where $Q_X$ is a nonnegative $|\mathcal{M}|$-dimensional vector which may not be a probability vector, and $P_{\bar{X}}(x) = Q_X(x)/\sum_{x' \in \mathcal{X}}Q_X(x')$. Then, for each $a \in \mathcal{X}$, define the partial derivatives of $R_{P_X}(d_s,d_x)$ with respect to $P_X(a)$ as
\begin{equation}
\dot{R}_X(a, d_s,d_x) \triangleq \frac{\partial}{\partial Q_{X}(a)}R_{Q_X}(d_s,d_x)\bigg|_{Q_X=P_{X}}.
\end{equation}

%In order to rigorously define a derivative of $R_{P_X}(d_s,d_x)$ with respect to $P_X$, we consider a finite measure on $\mathcal{X}$, $Q_X$, which is not necessarily a probability measure, and extend the definition of $R_{P_X}(d_s,d_x)$ to $Q_X$ as follows:
%\begin{equation}
%R_{Q_X}(d_s,d_x) \triangleq R_{P_{\bar{X}}}(d_s,d_x),
%\end{equation}
%where $P_{\bar{X}}(x) = Q_X(x)/\sum_{x \in \mathcal{X}}Q_X(x)$. For $a \in \mathcal{X}$, denote the partial derivatives 
%\begin{equation}
%\dot{R}_X(a, d_s,d_x) \triangleq \frac{\partial}{\partial Q_{X}(a)}R_{Q_X}(d_s,d_x)|_{Q_X=P_{X}}.
%\end{equation}

\begin{lemma} 
	\label{finite_alphabet_derivative}
	Fix $(d_s,d_x) \in \mathcal{D}_{\mathrm{in}}$. Assume that the alphabet $\mathcal{X}$ is finite and for all $P_{\bar{X}}$ in some neighborhood of $P_X$, $\textsf{supp}(P_{\bar{Z}^\star \bar{Y}^\star}) = \textsf{supp}(P_{Z^\star Y^\star})$, where $P_{\bar{Z}^\star \bar{Y}^\star}$ achieves $R_{\bar{X}}(d_s,d_x)$.
	%	\bluemark{$\lambda_s^\star$ and $\lambda_{s, \bar{X}} \triangleq -\partial R_{\bar{X}}(d_s,d_x)/\partial d_s$ are either both zero or both non-zero, and $\lambda_x^\star$ and $\lambda_{x, \bar{X}}\triangleq -\partial R_{\bar{X}}(d_s,d_x)/\partial d_x$ are also either both zero or both non-zero.} 
%	\footnote{\bluemark{This assumption implies that there exists a neighborhood of $P_X$ such that for any $P_{\bar{X}}$ in it, $\mathbb{E}\left[\bar{\textsf{d}}_s(X,Z^\star)\right] < d_s$ implies $\mathbb{E}\left[\bar{\textsf{d}}_s(\bar{X},\bar{Z}^\star)\right] < d_s$ and $\mathbb{E}\left[\textsf{d}_x(X,Y^\star)\right] < d_x$ implies $\mathbb{E}\left[\textsf{d}_x(\bar{X},\bar{Y}^\star)\right] < d_x$, where $(S, \bar{X})\sim P_{\bar{X}}P_{S|\bar{X}}$ and $P_{\bar{Z}^{\star}\bar{Y}^{\star}|\bar{X}}$ achieves $R_{S,\bar{X}}(d_s,d_x)$.}}
	Then
	\begin{align}
	\dot{R}_X(a, d_s,d_x) =& \jmath_{X}(a,d_s,d_x) - R_X(d_s,d_x), \label{Prop2_c1} \\
	\textrm{Var}\left[\dot{R}_X(X, d_s,d_x)\right] =& \textrm{Var}\left[\jmath_{X}(X,d_s,d_x)\right]. \label{Prop2_c2}
	\end{align}
	%	where $R_{\bar{X}}(d_s,d_x) = I(\bar{X}; \bar{Z}^\star, \bar{Y}^\star)$.
	
\end{lemma}
\begin{proof}[Proof]
	Recall that $\mathcal{D}_{\mathrm{in}} \triangleq \textsf{int}(\mathcal{D}_{sx}) \cup \textsf{int}(\mathcal{D}_{\bar{s}x}) \cup \textsf{int}(\mathcal{D}_{s\bar{x}}) \cup \textsf{int}(\mathcal{D}_{\bar{s}\bar{x}})$.
	Since, in the cases when $(d_s,d_x) \in \textsf{int}(\mathcal{D}_{\bar{s}x}) \cup \textsf{int}(\mathcal{D}_{s\bar{x}})$, both $R_X(d_s,d_x)$ and $\jmath_{X}(x,d_s,d_x)$ degenerate into scenarios previously addressed in \cite{Kostina2016Nonasymptotic}, where \eqref{Prop2_c1} and \eqref{Prop2_c2} have already been established, we now turn our attention to the remaining two cases.
	
	We first consider the case when $(d_s,d_x) \in \textsf{int}(\mathcal{D}_{\bar{s}\bar{x}})$. Define $d_{s,\max} \triangleq \min_{z \in \widehat{\mathcal{M}}} \mathbb{E}[\bar{\textsf{d}}_s(X,z)]$ and $d_{x,\max} \triangleq \min_{y \in \widehat{\mathcal{X}}} \mathbb{E}[\textsf{d}_x(X,y)]$. We note that
	\begin{align}
		\textsf{int}(\mathcal{D}_{\bar{s}\bar{x}}) = \{(d_s,d_x): d_s > d_{s,\max}, d_x> d_{x, \max}\},
	\end{align}
	and for all $(d_s,d_x) \in \textsf{int}(\mathcal{D}_{\bar{s}\bar{x}})$, we have $\lambda_s^\star = 0$, $\lambda_x^\star=0$, and $R_X(d_s,d_x) = 0$. Consequently, the right-hand side of \eqref{Prop2_c1} equals to $0$.
	Since $d_{s,\max}$ and $d_{x,\max}$ are both continuous functions of the distribution of $X$, there exists a neighborhood of $P_X$ such that for all $P_{\bar{X}}$ in it, we have $d_s > d_{s,\max}(P_{\bar{X}})$ and $d_x> d_{x, \max}(P_{\bar{X}})$ still hold, implying $R_{\bar{X}}(d_s,d_x) = 0$ in this neighborhood. Accordingly, the left-hand side of \eqref{Prop2_c1} also equals to $0$. This leads to the conclusion that \eqref{Prop2_c1} and \eqref{Prop2_c2} hold for $(d_s,d_x) \in \textsf{int}(\mathcal{D}_{\bar{s}\bar{x}})$. 
%	\eqref{Prop2_c2} is an immediate corollary to \eqref{Prop2_c1}.
	
%	First we show that \bluemark{$\lambda_{s, \bar{X}} \triangleq -\partial R_{\bar{X}}(d_s,d_x)/\partial d_s = 0$ if $\lambda_s^\star = 0$, and $\lambda_{x, \bar{X}}\triangleq -\partial R_{\bar{X}}(d_s,d_x)/\partial d_x = 0$ if $\lambda_x^\star = 0$.}
	
	We now consider the case when $(d_s,d_x) \in \textsf{int}(\mathcal{D}_{sx})$. First, we have \eqref{Prop2_step0}-\eqref{Prop2_step5}, 
	\begin{figure*}[!t]
		\begin{align}
		&\frac{\partial}{\partial Q_X(a)}\mathbb{E}[\jmath_{\bar{X}}(X,d_s,d_x)]\Big|_{Q_X=P_X} \label{Prop2_step0} \\
		=&\frac{\partial}{\partial Q_X(a)}\mathbb{E}[\imath_{\bar{X};\bar{Z}^\star \bar{Y}^\star}(X;Z^\star,Y^\star)] + \lambda_{s, \bar{X}} \mathbb{E}[\bar{\textsf{d}}_s(X,Z^\star) - d_s] + \lambda_{x, \bar{X}} \mathbb{E}[\textsf{d}_x(X,Y^\star) - d_x]\Big|_{Q_X=P_X} \label{Prop2_step1} \\
		=&\frac{\partial}{\partial Q_X(a)}\mathbb{E}[\imath_{\bar{X};\bar{Z}^\star \bar{Y}^\star}(X;Z^\star,Y^\star)]\Big|_{Q_X=P_X} \label{Prop2_step2}\\
		=&\frac{\partial}{\partial Q_X(a)}\mathbb{E}[\log P_{\bar{X}| \bar{Z}^\star \bar{Y}^\star}(X; Z^\star, Y^\star)]\Big|_{Q_X=P_X} -
		\frac{\partial}{\partial Q_X(a)}\mathbb{E}[\log P_{\bar{X}}(X)]\Big|_{Q_X=P_X} \label{Prop2_step3}\\
		=&\log e \cdot \frac{\partial}{\partial Q_X(a)}\mathbb{E}\left[\frac{ P_{\bar{X}| \bar{Z}^\star \bar{Y}^\star}(X; Z^\star, Y^\star)}{P_{X| Z^\star Y^\star}(X; Z^\star, Y^\star)}  \right]\Big|_{Q_X=P_X} - \log e \cdot \frac{\partial}{\partial Q_X(a)}\mathbb{E}\left[ \frac{P_{\bar{X}}(X)}{P_{X}(X)} \right]\Big|_{Q_X=P_X} \label{Prop2_step4}\\
		=&0 \label{Prop2_step5}
		\end{align}
		\hrulefill
	\end{figure*}
%	\begin{align}
%	&\frac{\partial}{\partial Q_X(a)}\mathbb{E}[\jmath_{\bar{X}}(X,d_s,d_x)]\Big|_{Q_X=P_X}\\
%	=&\frac{\partial}{\partial Q_X(a)}\mathbb{E}[\imath_{\bar{X};\bar{Z}^\star \bar{Y}^\star}(X;Z^\star,Y^\star)]  \label{Prop2_step1} \\ 
%	&+\! \lambda_{s, \bar{X}} \mathbb{E}[\bar{\textsf{d}}_s(X,Z^\star) \!-\! d_s] \!+\! \lambda_{x, \bar{X}} \mathbb{E}[\textsf{d}_x(X,Y^\star) \!-\! d_x]\Big|_{Q_X=P_X} \nonumber \\
%	=&\frac{\partial}{\partial Q_X(a)}\mathbb{E}[\imath_{\bar{X};\bar{Z}^\star \bar{Y}^\star}(X;Z^\star,Y^\star)]\Big|_{Q_X=P_X} \label{Prop2_step2}\\
%	=&\frac{\partial}{\partial Q_X(a)}\mathbb{E}[\log P_{\bar{X}| \bar{Z}^\star \bar{Y}^\star}(X; Z^\star, Y^\star)]\Big|_{Q_X=P_X} \nonumber \\ 
%	&-
%	\frac{\partial}{\partial Q_X(a)}\mathbb{E}[\log P_{\bar{X}}(X)]\Big|_{Q_X=P_X} \label{Prop2_step3}\\
%	=&\log e \cdot \frac{\partial}{\partial Q_X(a)}\mathbb{E}\left[\frac{ P_{\bar{X}| \bar{Z}^\star \bar{Y}^\star}(X; Z^\star, Y^\star)}{P_{X| Z^\star Y^\star}(X; Z^\star, Y^\star)}  \right]\Big|_{Q_X=P_X} \nonumber \\ 
%	&- \log e \cdot \frac{\partial}{\partial Q_X(a)}\mathbb{E}\left[ \frac{P_{\bar{X}}(X)}{P_{X}(X)} \right]\Big|_{Q_X=P_X} \label{Prop2_step4}\\
%	=&0, \label{Prop2_step5}
%	\end{align}
	where \eqref{Prop2_step1} is by \eqref{d_tilted_information_of_surrogate} and the assumption $\textrm{supp}(P_{\bar{Z}^\star \bar{Y}^\star}) = \textrm{supp}(P_{Z^\star Y^\star})$, 
	\eqref{Prop2_step2} holds since $\mathbb{E}[\bar{\textsf{d}}_s(X,Z^\star) - d_s] =0$ and $\mathbb{E}[\textsf{d}_x(X,Y^\star) - d_x]=0$ for $(d_s,d_x) \in \textsf{int}(\mathcal{D}_{sx})$,
%	by the assumptions and the complementary slackness condition, $\partial \lambda_{s, \bar{X}}/\partial Q_X(a) = 0$ if $\mathbb{E}[\bar{\textsf{d}}_s(X,Z^\star) - d_s] \neq 0$ and $\partial \lambda_{x, \bar{X}}/\partial Q_X(a) = 0$ if $\mathbb{E}[\textsf{d}_x(X,Y^\star) - d_x] \neq 0$, 
	and \eqref{Prop2_step5} holds since the two expectation terms in \eqref{Prop2_step4} are always equal to $0$. By \eqref{Prop2_step5}, we have
	\begin{align}
	&\dot{R}_X(a, d_s,d_x)\\ =&\frac{\partial}{\partial Q_{X}(a)}\mathbb{E}[\jmath_{\bar{X}}(\bar{X},d_s,d_x)]\bigg|_{Q_X=P_{X}}\\
	=&\jmath_{X}\!(a,d_s,d_x) \!-\! R_X\!(d_s,d_x) \!+\! \frac{\partial}{\partial Q_{X}(a)}\!\mathbb{E}[\jmath_{\bar{X}}\!(X,d_s,d_x)]\bigg|_{Q_X=P_{X}}\\
	=&\jmath_{X}(a,d_s,d_x) - R_X(d_s,d_x),
	\end{align}
	completing the proof of \eqref{Prop2_c1}. \eqref{Prop2_c2} is an immediate corollary to \eqref{Prop2_c1}.
\end{proof}

The Berry–Ess\'een central limit theorem (CLT) introduced in the following is of fundamental importance in our second-order analysis.
\begin{lemma} 
	\label{Berry_Esseen}
	(Berry–Ess\'een CLT, e.g., \cite[Theorem 13]{Kostina2012Fixed}, \cite[Chapter XVI.5, Theorem 2]{Feller1971Introduction}):
	Fix an integer $k > 0$. Let random variables $\{W_i \in \mathbb{R}\}_{i=1}^k $ be independent. Denote
	\begin{align}
	\mu_k =& \frac{1}{k}\sum_{i=1}^k \mathbb{E}[W_i],\\
	V_k =& \frac{1}{k}\sum_{i=1}^k \textrm{Var}[W_i],\\
	T_k =&\frac{1}{k}\sum_{i=1}^k \mathbb{E}\left[|W_i - \mathbb{E}[W_i]|^3\right],\\
	B_k =&6 \frac{T_k}{V_k^{3/2}}.
	\end{align}
	Then, for any $t > 0$,
	\begin{align}
	\label{P_Berry_Esseen}
	\left|\mathbb{P}\left[\sum_{i=1}^k W_i > k \left(\mu_k + t\sqrt{\frac{V_k}{k}}\right)\right] - Q(t)\right| \leq \frac{B_k}{\sqrt{k}},
	\end{align}
	where $Q(t)$ denotes the complementary standard Gaussian cumulative distribution function.
\end{lemma}

\section{Proof of the Converse Part of Theorem \ref{theorem_Gaussian_approximation}} \label{proof_converse_theorem_Gaussian_approximation}

If $\lambda_s^\star = 0$, by Corollary \ref{converse_bound_corollary2}, Theorem \ref{theorem_Gaussian_approximation} can be derived directly through the proof methodology employed for the converse part of \cite[Theorem 12]{Kostina2012Fixed}. In the following, we consider the case $\lambda_s^\star > 0$.

Our subsequent proof is adapted from the converse proof of \cite[Theorem 5]{Kostina2016Nonasymptotic}. Let
\begin{equation}
\label{log_M}
\begin{aligned}
\log M = &k R_X(d_s, d_x) + \sqrt{k \tilde{\mathcal{V}}(d_s,d_x)} Q^{-1}(\epsilon_k)\\ 
&- \frac{1}{2} \log k - \log |\mathcal{P}_{[k]}| - c |\mathcal{A}|\log k,
\end{aligned}
\end{equation}
where $\epsilon_k = \epsilon + O\left(\frac{\log k}{\sqrt{k}}\right)$, $c$ is a constant that will be specified in the sequel, and $\mathcal{P}_{[k]}$ denotes the set of all conditional $k$-types $\hat{\mathcal{S}} \times \hat{\mathcal{A}} \to \mathcal{A}$. By type counting, we have $|\mathcal{P}_{[k]}| \leq (k+1)^{|\mathcal{A}| |\hat{\mathcal{S}}|  |\hat{\mathcal{A}}|}$.

We weaken the bound in Corollary \ref{converse_bound_corollary} by choosing
\begin{gather}
P_{\!\bar{X}^k|\bar{Z}^k=z^k,\bar{Y}^k=y^k}\!(x^k)\!=\!\frac{1}{|\mathcal{P}_{[k]}|}\!\!\sum_{P_{X\!|\!ZY} \in \mathcal{P}_{[k]}}\! \prod_{i=1}^k\! P_{X|Z=z_i, Y\!=y_i} \!(x_i), \label{sum_P_X_ZY} \\
\lambda_s= k \lambda_s(x^k) = -k\frac{\partial R_{\textrm{type}(x^k)}(d_s,d_x)}{\partial d_s}\\
\lambda_x= k \lambda_x(x^k) = -k\frac{\partial R_{\textrm{type}(x^k)}(d_s,d_x)}{\partial d_x}\\
\gamma=\frac{1}{2} \log k. \label{gamma_value}
\end{gather}
By Corollary \ref{converse_bound_corollary}, any $(k, M, d_s, d_x, \epsilon')$ code with $M$ given in \eqref{log_M} must satisfy
\begin{align} 
\label{block_converse_bound}
\epsilon' \geq& \mathbb{E}\Big[\min_{z^k \in \hat{\mathcal{S}}^k, y^k \in \hat{\mathcal{A}}^k}
\mathbb{P}\Big[\imath_{\bar{X}^k|\bar{Z}^k \bar{Y}^k\Vert X^k}(X^k;z^k,y^k)\nonumber\\
&+ k \lambda_s(X^k)(\textsf{d}_s(S^k,z^k) - d_s)\nonumber\\
&+ k \lambda_x(X^k) (\textsf{d}_x(X^k,y^k) - d_x) \geq \log M + \gamma \Big| X^k \Big]\Big]\nonumber\\ 
&- \exp(-\gamma).
\end{align}

Below, we prove that the right side of \eqref{block_converse_bound} is lower-bounded by $\epsilon$ for $M$ in \eqref{log_M}, implying that the logarithm of $M$ in any $(k, M, d_s, d_x, \epsilon)$ code must be no less than the right side of \eqref{log_M}. For each triple $(x^k, z^k, y^k)$, abbreviate
\begin{align}
\textrm{type}\left(x^k\right)=&P_{\bar{X}}, \\
\textrm{type}\left(z^k, y^k|x^k\right)=&P_{\bar{Z}\bar{Y}|\bar{X}}, \\
\lambda_s(x^k)=& \lambda_{s,\bar{X}}, \\
\lambda_x(x^k)=& \lambda_{x,\bar{X}},
\end{align}
and define independent random variables
\begin{equation}
	\begin{aligned}
	W_i \triangleq& I(\bar{X}; \bar{Z}, \bar{Y}) + \lambda_{s, \bar{X}}(\textsf{d}_s(S_i, z_i) - d_s)\\ 
	&+ \lambda_{x, \bar{X}}(\textsf{d}_x(x_i, y_i) - d_x),\ i=1,\dots,k,
	\end{aligned}
\end{equation}
where $S_i \sim P_{S|X=x_i}$, $i=1,\dots,k$. Since 
\begin{gather}
\mathbb{E}\!\Big[\sum_{i=1}^k \textsf{d}_s(S_i, z_i)\Big]=k \mathbb{E}\!\left[\mathbb{E}\left[\textsf{d}_s\!(S, \bar{Z})|\bar{X}, \bar{Z}\right]\right]\!=\!k \mathbb{E}\left[\textsf{d}_s\!(S, \bar{Z})\right]\!, \\
\sum_{i=1}^k \textsf{d}_x(x_i, y_i) =k\mathbb{E}\left[\textsf{d}_x(\bar{X}, \bar{Y})\right], \\
\begin{aligned}
&\text{Var}\Big[\sum_{i=1}^k \textsf{d}_s(S_i, z_i)\Big]\\
=& k\mathbb{E}\big[\mathbb{E}\big[\big(\textsf{d}_s(S, \bar{Z})- \mathbb{E}\big[\textsf{d}_s(S, \bar{Z})|\bar{X}, \bar{Z}\big]\big)^2| \bar{X}, \bar{Z}\big]\big]\\
=& k\textrm{Var}\left[\textsf{d}_s(S, \bar{Z})|\bar{X}, \bar{Z}\right],
\end{aligned}
\end{gather}
where $P_{S\bar{X}\bar{Z}\bar{Y}} = P_{S|X}P_{\bar{X}\bar{Z}\bar{Y}}$,
in the notation of \cite[Theorem 11]{Kostina2016Nonasymptotic}, we have
\begin{align}
\mu_k(P_{\bar{Z}\bar{Y}|\bar{X}})=&I(\bar{X};\bar{Z}, \bar{Y}) + \lambda_{s, \bar{X}} \left(\mathbb{E}\left[\textsf{d}_s(S,\bar{Z})\right] - d_s\right)\nonumber \\ 
&+ \lambda_{x, \bar{X}} \left(\mathbb{E}\left[\textsf{d}_x(\bar{X},\bar{Y})\right] - d_x\right), \\
V_k(P_{\bar{Z}\bar{Y}|\bar{X}}) =&\lambda_{s,\bar{X}}^2 \textrm{Var}\left[\textsf{d}_s(S, \bar{Z})|\bar{X}, \bar{Z}\right], \\
T_k(P_{\bar{Z}\bar{Y}|\bar{X}})=& \lambda_{s,\bar{X}}^3 \mathbb{E}\!\left[\!\left|\textsf{d}_s(S, \bar{Z}) \!-\! \mathbb{E}\!\left[\textsf{d}_s(S,\bar{Z})|\bar{X}, \bar{Z}\right]\!\right|^3\!\right]\!\!.
\end{align}
%We define $P_{\bar{X}|\bar{Z} \bar{Y}}$ through $P_{\bar{X}} P_{\bar{Z}\bar{Y}|\bar{X}}$ and lower-bound the sum in \eqref{sum_P_X_ZY} by the term containing $P_{\bar{X}|\bar{Z}\bar{Y}}$, concluding that

Besides, we can write
\begin{align}
&\imath_{\bar{X}^k|\bar{Z}^k \bar{Y}^k\Vert X^k}(x^k;z^k,y^k)
+ k \lambda_s(x^k)(\textsf{d}_s(S^k,z^k) - d_s)\nonumber\\
&+ k \lambda_x(x^k) (\textsf{d}_x(x^k,y^k) - d_x)\\
\geq&kI(\bar{X};\bar{Z}, \bar{Y}) + kD(\bar{X}\Vert X) + \lambda_{s, \bar{X}} \Big(\sum_{i=1}^k\textsf{d}_s(S_i,z_i) - k d_s\Big)\nonumber\\ 
&+ \lambda_{x, \bar{X}} \Big(\sum_{i=1}^k\textsf{d}_x(x_i,y_i) - k d_x\Big) - \log |\mathcal{P}_{[k]}| \label{lower_bound_i1} \\
=&\sum_{i=1}^k W_i + kD(\bar{X}\Vert X) - \log |\mathcal{P}_{[k]}|\\
\geq&\sum_{i=1}^k W_i - \log |\mathcal{P}_{[k]}|, \label{lower_bound_Wi}
\end{align}
where \eqref{lower_bound_i1} is by lower-bounding the sum in \eqref{sum_P_X_ZY} using the term containing $P_{\bar{X}|\bar{Z}\bar{Y}}$, which is defined through $P_{\bar{X}} P_{\bar{Z}\bar{Y}|\bar{X}}$.
Weakening \eqref{block_converse_bound} using \eqref{lower_bound_Wi}, we have
\begin{align} 
\label{block_converse_bound2}
\epsilon' \!\geq& \mathbb{E}\Big[\!\min_{z^k \in \hat{\mathcal{S}}^k, y^k \in \hat{\mathcal{A}}^k}\!
\mathbb{P}\Big[\sum_{i=1}^k W_i \geq \log M + \gamma + \log |\mathcal{P}_{[k]}| \Big| X^k  \Big]\Big]\nonumber\\ 
&- \exp(-\gamma).
\end{align}

Define the typical set of $x^k$, $\mathcal{T}_k$, as
\begin{equation}
\label{definition_typical_set}
\mathcal{T}_k \triangleq \left\{x^k \in \mathcal{A}^k: \|\textrm{type}(x^k) - P_X\|^2 \leq \left|\mathcal{A}\right|\frac{\log k}{k}\right\},
\end{equation}
where $\|\cdot\|$ denotes the Euclidean norm. Note that by \cite[Lemma 1]{Kostina2016Nonasymptotic}, 
\begin{equation}
\label{Xk_doesnot_typical}
\mathbb{P}\left[X^k \notin \mathcal{T}_k \right] \leq \frac{2 |\mathcal{A}|}{\sqrt{k}}.
\end{equation}
Next, we evaluate the minimum in \eqref{block_converse_bound2} for $x^k \in \mathcal{T}_k$.

If $V_k(P_{Z^\star Y^\star|X}) = 0$, we have $\textsf{d}_s(S,z) - \bar{\textsf{d}}_s(x,z) \overset{a.s.}{=} 0$ for all $z \in \textrm{supp}(P_{Z^\star})$. Therefore, the considered problem simplifies to a noiseless two-constraint source coding problem \cite[Section VI]{Blahut1972Computation}, \cite[Problem 10.19]{Cover2006Elements}, allowing us to establish our second-order result in a manner analogous to the proof of \cite[Theorem 12]{Kostina2012Fixed}. In the following, we assume that $V_k(P_{Z^\star Y^\star|X}) > 0$. Similar to the argument in the converse proof of \cite[Theorem 5]{Kostina2016Nonasymptotic},
%\cite[(C.50)]{Kostina2013Lossy},
the conditions of \cite[Theorem 11]{Kostina2016Nonasymptotic} are satisfied by $\{W_i\}_{i=1}^k$, with $\mu_k^\star = \mu_k(P_{\bar{Z}^\star\bar{Y}^\star| \bar{X}})$ and $V_k^\star = V_k(P_{\bar{Z}^\star\bar{Y}^\star| \bar{X}})$. Therefore, denoting
\begin{equation}
\Delta_k(P_{\bar{X}}) = \log M + \gamma + \log |\mathcal{P}_{[k]}| - k R_{\bar{X}}(d_s,d_x),
\end{equation}
where $M$ and $\gamma$ are chosen in \eqref{log_M} and \eqref{gamma_value}, by \cite[Theorem 11]{Kostina2016Nonasymptotic}, we have
\begin{align}
\label{Q_lower_bound_1}
&\min_{P_{\bar{Z}\bar{Y}|\bar{X}}} \mathbb{P}\Big[\sum_{i=1}^k W_i \geq \log M + \gamma + \log |\mathcal{P}_{[k]}| \Big| \textrm{type}(X^k) = P_{\bar{X}} \Big]\nonumber\\ 
&\geq Q\Bigg(\frac{\Delta_k(P_{\bar{X}})}{ \lambda_{s,\bar{X}}\sqrt{k \textrm{Var}\left[\textsf{d}_s(S, \bar{Z}^{\star})|\bar{X}, \bar{Z}^{\star}\right]}}\Bigg) - \frac{K}{\sqrt{k}},
\end{align}
where $K > 0$ is that in \cite[Theorem 11]{Kostina2016Nonasymptotic}.

By assumption \ref{differentiability} and \eqref{Prop2_c2}, we can apply a Taylor series expansion in a neighborhood of $P_X$ to $\frac{1}{ \lambda_{s,\bar{X}}\sqrt{\textrm{Var}\left[\textsf{d}_s(S, \bar{Z}^{\star})|\bar{X}, \bar{Z}^{\star}\right]}}$. Consequently, for some scalars $a$ and $K_1 \geq 0$, we have
\begin{align}
&Q\left(\frac{\Delta_k(P_{\bar{X}})}{ \lambda_{s,\bar{X}}\sqrt{k \textrm{Var}\left[\textsf{d}_s(S, \bar{Z}^{\star})|\bar{X}, \bar{Z}^{\star}\right]}}\right)\\
\bluemark{\geq}& Q\left(\frac{\Delta_k(P_{\bar{X}})}{ \lambda_{s,X}\sqrt{k \textrm{Var}\left[\textsf{d}_s(S, Z^\star)|X, Z^\star \right]}}\left(1 + a \sqrt{\frac{\log k}{k}}\right)\right)\\
\geq& Q\left(\frac{\Delta_k(P_{\bar{X}})}{ \lambda_{s,X}\sqrt{k \textrm{Var}\left[\textsf{d}_s(S, Z^\star)|X, Z^\star \right]}}\right) - K_1 \frac{\log k}{\sqrt{k}}, \label{Q_lower_bound1}
\end{align}
where \eqref{Q_lower_bound1} holds since
$Q(x + \xi) \geq Q(x) - \frac{|\xi|^+}{\sqrt{2 \pi}}$
and $\Delta_k(P_{\bar{X}}) = O(\sqrt{k \log k})$ for $x^k \in \mathcal{T}_k$.

Besides, for $x^k \in \mathcal{T}_k$, there exists $c > 0$ such that
\begin{align}
&R_{\bar{X}}(d_s,d_x) \nonumber \\ 
\geq& R_{X}(d_s,d_x) + \sum_{a \in \mathcal{A}}(P_{\bar{X}}(a) - P_{X}(a)) \dot{R}_X(a, d_s,d_x)\nonumber\\ 
&- c \|P_{\bar{X}} - P_X\|^2 \label{Taylor_apply} \\
=& R_{X}(d_s,d_x) + \frac{1}{k} \sum_{i=1}^k \dot{R}_X(x_i, d_s, d_x)\nonumber\\ 
&- \mathbb{E}\left[\dot{R}_X(X, d_s, d_x)\right] - c \|P_{\bar{X}} - P_X\|^2\\
=& \mathbb{E}\left[\jmath_{X}(\bar{X},d_s,d_x)\right] - c \|P_{\bar{X}} - P_X\|^2 \label{app_Lemma1} \\
\geq& \mathbb{E}\left[\jmath_{X}(\bar{X},d_s,d_x)\right] - c |\mathcal{A}| \frac{\log k}{k}, \label{R_X_bar_lower_bound}
\end{align}
where \eqref{Taylor_apply} is by Taylor's theorem, applicable because assumption \ref{differentiability} holds, \eqref{app_Lemma1} is by Lemma \ref{finite_alphabet_derivative}, and \eqref{R_X_bar_lower_bound} is by the definition \eqref{definition_typical_set}. Then, by introducing independent random variable $G \sim \mathcal{N}(0, 1)$, we have
\begin{align}
&Q\left(\frac{\Delta_k(P_{\bar{X}})}{ \lambda_{s,X}\sqrt{k \textrm{Var}\left[\textsf{d}_s(S, Z^\star)|X, Z^\star \right]}}\right)\\
=& \mathbb{P}\Big[k R_{\bar{X}}(d_s,d_x) + \lambda_{s,X}\sqrt{k \textrm{Var}\left[\textsf{d}_s(S, Z^\star)|X, Z^\star \right]} G\nonumber\\ 
&\geq \log M + \gamma + \log |\mathcal{P}_{[k]}| \Big]\\
\geq& \mathbb{P}\Big[k\mathbb{E}\left[\jmath_{X}(\bar{X},d_s,d_x)\right] + \lambda_{s,X}\sqrt{k \textrm{Var}\left[\textsf{d}_s(S, Z^\star)|X, Z^\star \right]} G\nonumber\\ 
&\geq \log M + a_k\Big], \label{lower_bound_Q_2}
\end{align}
where $a_k \triangleq \gamma + \log |\mathcal{P}_{[k]}| + c |\mathcal{A}| \log k$.

Finally, we have \eqref{step1}-\eqref{step7},
\begin{figure*}[!t]
\begin{align}
\epsilon'
\geq& \mathbb{E}\left[\min_{z^k \in \hat{\mathcal{S}}^k, y^k \in \hat{\mathcal{A}}^k}
\mathbb{P}\left[\sum_{i=1}^k W_i \geq \log M + \gamma + \log |\mathcal{P}_{[k]}| \bigg| X^k  \right]\right] - \exp(-\gamma) \label{step1}\\
=&\mathbb{E}\left[\min_{P_{\bar{Z}\bar{Y}|\bar{X}}} \mathbb{P}\left[\sum_{i=1}^k W_i \geq \log M + \gamma + \log |\mathcal{P}_{[k]}| \bigg| \textrm{type}(X^k) = P_{\bar{X}} \right]\right] - \exp(-\gamma) \label{step2}\\
\geq&\mathbb{E}\left[\min_{P_{\bar{Z}\bar{Y}|\bar{X}}} \mathbb{P}\left[\sum_{i=1}^k W_i \geq \log M + \gamma + \log |\mathcal{P}_{[k]}| \bigg| \textrm{type}(X^k) = P_{\bar{X}} \right] \textbf{1}\left\{X^k \in \mathcal{T}_k\right\}\right] - \exp(-\gamma) \label{step3}\\
\geq&\mathbb{E}\!\left[\mathbb{P}\!\left[k\mathbb{E}\left[\jmath_{X}\!(\bar{X},d_s,d_x)\right] \!+\! \lambda_{s,X}\!\sqrt{k \textrm{Var}\left[\textsf{d}_s(S, Z^\star)|X, Z^\star \right]} G \!\geq\! \log M \!+\! a_k\right] \!\textbf{1}\!\left\{X^k \!\in\! \mathcal{T}_k\right\}\!\right] \!-\! \frac{K}{\sqrt{k}} \!-\! K_1 \frac{\log k}{\sqrt{k}} \!-\! \exp(-\gamma) \label{step4}\\
\geq&\mathbb{P}\!\left[\sum_{i=1}^k\jmath_{X}(X_i,d_s,d_x) \!+\! \lambda_{s,X}\sqrt{k \textrm{Var}\left[\textsf{d}_s(S, Z^\star)|X, Z^\star \right]} G \!\geq\! \log M \!+\! a_k\right] \!-\! \mathbb{P}\left[X^k \notin \mathcal{T}_k\right]\!-\! \frac{K}{\sqrt{k}} \!-\! K_1 \frac{\log k}{\sqrt{k}} \!-\! \exp(-\gamma) \label{step5}\\
\geq&\epsilon_k - \frac{B}{\sqrt{k+1}} - \mathbb{P}\left[X^k \notin \mathcal{T}_k\right] 
- \frac{K}{\sqrt{k}} - K_1 \frac{\log k}{\sqrt{k}} - \exp(-\gamma) \label{step6}\\
\geq& \epsilon_k - \frac{B + 2|\mathcal{A}| + K + K_1 \log k + 1}{\sqrt{k}} \label{step7}
\end{align}
	\hrulefill
\end{figure*}
where 
\begin{itemize}
	\item \eqref{step1} is by \eqref{block_converse_bound2};
	\item \eqref{step2} is by the observation that $\sum_{i=1}^k W_i$ remains unchanged for fixed $P_{\bar{X}\bar{Z}\bar{Y}}$;
	\item \eqref{step4} is by \eqref{Q_lower_bound_1}, \eqref{Q_lower_bound1} and \eqref{lower_bound_Q_2};
	\item \eqref{step5} is by the union bound and the definition of $\mathbb{E}\left[\jmath_{X}(\bar{X},d_s,d_x)\right]$;
	\item \eqref{step6} is by the Berry-Ess\'een theorem (Lemma \ref{Berry_Esseen}), the choice of $M$ in \eqref{log_M}, and Proposition \ref{Prop_relationship_V_tildeV}, where $B$ is that in Lemma \ref{Berry_Esseen};
	\item \eqref{step7} is by the choice of $\gamma$ in \eqref{gamma_value} and property \eqref{Xk_doesnot_typical}.
\end{itemize}
Our proof completes by letting
\begin{equation}
\epsilon_k = \epsilon + \frac{B + 2|\mathcal{A}| + K + K_1 \log k + 1}{\sqrt{k}}.
\end{equation}

%\cite{Budkuley2017Coding}
\bibliographystyle{IEEEtran}
\bibliography{ref}

\end{document}